\documentclass[11pt, a4paper]{article}
\usepackage{algorithm}
\usepackage{algorithmicx}
\usepackage{algpseudocode}
\usepackage{amsfonts}
\usepackage{amsmath}
\usepackage{amssymb}
\usepackage{amsthm}
\usepackage[page]{appendix}
\usepackage{authblk}
\usepackage{booktabs}
\usepackage{float}
\usepackage{fullpage}
\usepackage[left=2cm,top=2cm,right=2cm,bottom=2cm,nohead,nofoot]{geometry}
\usepackage{graphicx}
\usepackage{setspace}
\usepackage{hyperref}
\usepackage[retainorgcmds]{IEEEtrantools}
\usepackage{latexsym}
\usepackage{lscape}
\usepackage{listings}
\usepackage{mathrsfs}
\usepackage{rotating}
\usepackage{url}
\begin{document}
\newtheorem{theorem}{Theorem}[section]
\newtheorem{lemma}[theorem]{Lemma}
\newtheorem{corollary}[theorem]{Corollary}
\newtheorem{proposition}[theorem]{Proposition}
\newtheorem{claim}{Claim}
\theoremstyle{definition}
\newtheorem{definition}[theorem]{Definition}
\newtheorem{fact}{Fact}
\newtheorem{example}[theorem]{Example}
\newtheorem{conjecture}{Conjecture}
\newtheorem{xca}[theorem]{Exercise}
\newtheorem{openproblem}[conjecture]{Open Problem}
\theoremstyle{remark}
\newtheorem{remark}[theorem]{Remark}
\numberwithin{equation}{section}
\title{Using evolutionary computation to create vectorial Boolean functions with low differential uniformity and high nonlinearity.}
\author{James McLaughlin\footnote{Corresponding author, \nolinkurl{jmclaugh@cs.york.ac.uk}}, John A. Clark}
\date{}
\maketitle

\begin{abstract}
The two most important criteria for vectorial Boolean functions used as S-boxes in block ciphers are
differential uniformity and nonlinearity. Previous work in this field has focused only on
nonlinearity and a different criterion, autocorrelation. In this paper, we describe the results of
experiments in using simulated annealing, memetic algorithms, and ant colony optimisation to
create vectorial Boolean functions with low differential uniformity.

\textbf{Keywords:} Metaheuristics, simulated annealing, memetic algorithms, ant colony optimization, cryptography, Boolean functions, vectorial Boolean functions.
\end{abstract}

\thispagestyle{empty}

\section{Introduction}
Most block cipher designs rely on vectorial Boolean functions known as ``S-boxes''. These functions, mapping $n$ input bits to $m$ output bits, introduce nonlinearity into the cipher design, and are responsible for ensuring that no output bit of the cipher can be expressed as an easily-solved, low-degree polynomial in the input bits.

There have been various approaches to designing these functions, and most of the research in this area has focused on bijective S-box designs. Nyberg \cite{Nyberg1993} treated the $n$-bit input $(x_{1}, x_{2}, {\ldots}, x_{n})$ as a single value $x$ in the finite field $GF(2^n)$, and studied the properties of mappings $x^p$, where $p$ was either a fixed value or a polynomial of some particular form. The mapping $S(0)=0$, $S(x \neq 0) = x^{-1}$, known as the ``inverse based'' S-box, was used in the Advanced Encryption Standard \cite{FIPS-197} due to its excellent differential uniformity, nonlinearity and other properties.

The designers of Serpent \cite{Anderson_Biham_Knudsen_postAug1998}, by contrast, started with a set of $4 \times 4$ (i.e. 4 input bits, 4 output bits) bijections and a list of specified properties, and continued making random changes to the bijections until they had obtained eight with the required properties. The designers of PRESENT \cite{Bogdanov_Knudsen_Leander_Paar_Poschmann_Robshaw_Seurin_Vikkelsoe2007}, due to technological advances in the nine years since then, were simply able to choose a $4 \times 4$ bijection from the complete list \cite{Leander_Poschmann2007} of all $4 \times 4$ bijections with these properties.

Although there had been some earlier work using simpler techniques such as hill-climbing, and although metaheuristics had been applied to the related problem of evolving single-output-bit Boolean functions \cite{Clark_Jacob2000, Clark_Maitra_Stanica2004, Clark_Jacob_Stepney_searchcost2004, Clark_Jacob_Maitra_Stanica2004, Clark_Jacob_Stepney_Maitra_Millan2002}, the first major attempt to construct S-boxes by utilising metaheuristics was carried out in 1999 by Millan et al. \cite{Burnett_Carter_aClark_Dawson_Millan1999}. They attempted to maximise one property, known as nonlinearity, and simultaneously to minimise another; autocorrelation. The metaheuristic used was a genetic algorithm variant without mutation and with an unconventional crossover function that the authors claimed introduced enough randomness to eliminate the need for mutation.

Encouraged by the success of their previous research into the use of simulated annealing to evolve Boolean functions with 1-bit output length and low autocorrelation/high nonlinearity \cite{Clark_Jacob2000, Clark_Jacob_Stepney_searchcost2004, Clark_Jacob_Maitra_Stanica2004} \cite{Clark_Jacob_Stepney_Maitra_Millan2002}, Clark, Jacob and Stepney employed the same techniques to try to obtain S-boxes with these qualities \cite{Clark_Jacob_Stepney2004}.
Their approach differed from that which had preceded it in two key respects:

\begin{itemize}
\item Instead of simply using the value of nonlinearity/autocorrelation for the cost/fitness function, the whole linear approximation/autocorrelation table would be taken into account, with the cost function focusing on trying to bring every value close to zero.
\item After using one cost function for the simulated annealing, they would end the search by hill-climbing, usually with a different cost function - so, for example, after obtaining a relatively flat linear approximation table, they would hill-climb with nonlinearity as the fitness function, based on the hypothesis that the primary cost function would have guided the search into a region of the search space in which a higher number of candidate solutions with above average nonlinearity existed.
\end{itemize}

This approach improved on the best nonlinearity obtained by Millan et al. for bijections such that $6 \leq n = m \leq 8$, and equalled it for $n=m=5$. However, the researchers noted that S-boxes with higher nonlinearity were known to exist for $n=6$ and $n=8$. They also noted that a slight increase in the value of $n$ resulted in a massive increase to the problem complexity.

\subsection{Technical background}
\begin{definition}
For some $n, m$, an $n \times m$ S-box is a mapping from $GF(2^{n})$ to $GF(2^{m})$. If every value $\in GF(2^{m})$ is mapped to by an equal number of distinct input values, we say that the S-box is \textit{balanced}.
\end{definition}

Although the original Data Encryption Standard \cite{FIPS-46-3} used S-boxes mapping a six-bit output to a four-bit input, most modern cipher designs use only bijective $n \times n$ S-boxes. In particular, we note that the current Advanced Encryption Standard (AES) \cite{FIPS-197}, and most block ciphers designed for environments where the AES is too resource-intensive (e.g. PRESENT \cite{Bogdanov_Knudsen_Leander_Paar_Poschmann_Robshaw_Seurin_Vikkelsoe2007}, PRINTCIPHER \cite{Knudsen_Leander_Poschmann_Robshaw2010}), use only bijective $n \times n$ S-boxes. For this reason, we will focus primarily on functions of this sort.

Even if it is not bijective, the S-box is usually required to be balanced. Otherwise, the two most important criteria for a cryptographically secure S-box are low \textit{differential uniformity} and high \textit{nonlinearity}. There are other relevant criteria as well - as an example, it has been suggested that if large numbers of low-degree implicit equations with the S-box's input and output bits as variables hold, this may lead to the cipher being broken \cite{Courtois_Pieprzyk2002FinalIACR, Mouha_Preneel_Sun_Wang2011}.

\begin{definition}\label{DDT_defn}
Let $S$ be an S-box. Let us construct a table with $2^{n}$ rows and $2^{m}$ columns by defining the entry in row $i$, column $j$ to be the number of inputs $x$ such that $S(x) {\oplus} S(x \oplus i) = j$. (${\oplus}$ will signify exclusive-or throughout this paper.)
This table is known as the \textit{difference distribution table}, or DDT.

Note that all entries in this table must be even, for if $S(x) {\oplus} S(x \oplus i) = j$, then let $y = (x {\oplus} i)$, and we see that $S(y) {\oplus} S(y {\oplus} i) = j$. Furthermore, it is clear that all entries in the row of a DDT must sum to $2^{n}$, and hence that the sum of all entries in all rows must be equal to $2^{m+n}$.
\end{definition}

\begin{definition}\label{differential_uniformity_defn}
The \textit{differential uniformity} of an S-box $S$, sometimes denoted $DU(S)$ (or just $DU$ when it is clear from the context that $S$ is meant), is the largest value present in its difference distribution table (barring the entry for $i=j=0$).

Where $DU(S) = k$, we refer to $S$ as being \textit{differentially $k$-uniform}.
\end{definition}

\begin{definition}\label{APN_defn}
An S-box with differential uniformity 2 is referred to as being \textit{almost perfect nonlinear} (or ``APN''). In theory, perfect nonlinearity would correspond to $DU\ 1$, however since all entries in the DDT must be even, this is not achievable in practice. The term "almost perfect nonlinear" was introduced by Vaudenay et al. \cite{Chabaud_Vaudenay1994} in 1994.
\end{definition}

\begin{definition}\label{differential_frequency_defn}
The number of times that the value $D$ equal to $DU(S)$ appears in the difference distribution table of $S$ is referred to as its \textit{differential frequency} ($DF$ for short).
\end{definition}

\begin{definition}\label{LAT_defn}
Let $S$ be an S-box. Let us construct another table with $2^{n}$ rows and $2^{m}$ columns by defining the entry in row $i$, column $j$ as follows: where the parity of $i\ AND\ x$ is equal to the parity of $j\ AND\ S(x)$ for $k$ values of $x$, the corresponding entry in the table is the value $k - 2^{n-1}$.

This table is known as the \textit{linear approximation table}, or LAT. The proof is a little too long to reproduce here, but as long as the S-box is balanced, this table also has the property that all its entries are even.
\end{definition}

\begin{definition}\label{nonlinearity_defn}
The \textit{nonlinearity} of an S-box, $NL(S)$ (or just $NL$ when the box involved is clear from the context.) is equal to $2^{n-1} - max_{i, j}|LAT_{ij}|$ - that is, $2^{n-1}$ minus the maximum absolute value in the LAT (excluding the entry for $i=j=0$.)
\end{definition}

\begin{definition}\label{nonlinear_frequency_defn}
The number of times that the value $2^{n-1} - NL(S)$ appears in the linear approximation table of $S$ is referred to as the \textit{nonlinear frequency} of $S$ ($NF$ for short).
\end{definition}

Differential uniformity and nonlinearity are the most important properties in a cipher's S-boxes, being crucial to its resistance against differential cryptanalysis \cite{Biham_Shamir1990full, Biham_Shamir1992} and linear cryptanalysis \cite{Matsui1993, Matsui1994} respectively, as well as the many variants on and hybrids of these techniques.

As stated earlier, previous research also focused on a property known as \textit{autocorrelation}. While this has some relevance to the cipher's resistance against the hybrid technique \textit{differential-linear cryptanalysis} \cite{Hellman_Langford1994, Biham_Dunkelman_Keller2002}, differential uniformity and nonlinearity are far more important even in this context. We did not therefore focus on autocorrelation in our experiments, but for the sake of completeness when discussing previous research we provide the relevant definitions here:

\begin{definition}\label{autoc_table_defn}
Where $S$ is an S-box, let us construct a third table with $2^{n}$ rows and $2^{m}$ columns. Let
$S_{j}(x)$ be the Boolean function with one output bit defined by xoring the output bits of $S$
corresponding to the 1s in the bitmask $j$. Let the entry in row $i$, column $j$ be equal to the
number of inputs for which $S_{j}(x) = S_{j}(x \oplus i)$, minus the number of inputs for which
this was not the case.

This table is the \textit{autocorrelation table}, or ACT.
\end{definition}

\begin{definition}\label{autocorrelation_defn}
The \textit{autocorrelation} of an S-box, $AC(S)$ (or just $AC$ when it is clear from the context
that $S$ is the box involved) is the maximum absolute value in its autocorrelation table, barring
the trivial entry for $i=j=0$.
\end{definition}

\begin{definition}
The number of times that the value $AC(S)$ appears in the autocorrelation table of $S$ is referred
to as the \textit{autocorrelation frequency} of $S$ ($AF$ for short).
\end{definition}

The three tables are in fact related. Using matrix transformations, the ACT can be derived from
the DDT, and a table $C$ such that $C_{i, j} = (2{\cdot}LAT_{i, j})^{2}$ can be derived from either
the LAT or the ACT \cite{Imai_Zhang_Zheng2000}. Since this will prove relevant later on, we explain it in more detail here.

\begin{definition}
The Walsh-Hadamard matrix is also known as the ``Sylvester-Hadamard matrix''. It is defined recursively thus:

\begin{equation*}
H_{0} =
\begin{bmatrix}
1
\end{bmatrix}
\end{equation*}

\begin{equation*}
H_{n} =
\begin{bmatrix}
H_{n-1} &   H_{n-1} \\
H_{n-1} & -1{\times}H_{n-1}
\end{bmatrix}
\end{equation*}

So, for example,
\begin{equation*}
H_{1} = 
\begin{bmatrix}
1 &  1 \\
1 & -1
\end{bmatrix}
\qquad
H_{2} =
\begin{bmatrix}
1 &  1 &  1 &  1 \\
1 & -1 &  1 & -1 \\
1 &  1 & -1 & -1 \\
1 & -1 & -1 &  1
\end{bmatrix}
\end{equation*}
\end{definition}

\begin{theorem}\label{hadamard_ddt_lat_theorem}
For an $n \times m$ S-box $S$, let $D_{S}$ denote the DDT of $S$. Treat it as a matrix for the purposes of multiplication; then $H_{n}D_{S}H_{m}$ is the aforementioned table $C$ such that $C_{i, j} = (2{\cdot}LAT_{i, j})^{2}$.
\end{theorem}

For a proof, the reader is referred to the paper in which this theorem originally appeared \cite{Imai_Zhang_Zheng2000}.

\begin{definition}\label{ANF_defn}
Let $f$ be a Boolean function mapping $n$ input bits $x_i$ to one output bit $y$. $f$ may be uniquely expressed as a multivariate polynomial in which the variables are the values $x_i$, the AND operation is used for multiplication, and the XOR operation is used for addition. This is its \textit{algebraic normal form} (ANF). By means of the so-called M{\" o}bius Transform, the truth table of $f$ may be converted to the ANF, and vice versa.
\end{definition}

\begin{definition}\label{degree_defn}
The \textit{algebraic degree} of a Boolean function $f$ with one output bit is defined thus: Consider the ANF of $f$. Let the Hamming weight of a monomial in the ANF be defined as the number of variables multiplied together in the monomial - so, for instance, $x_{1}x_{4}$ has weight 2. The algebraic degree of a monomial is defined as being equal to its weight, and the algebraic degree of $f$ is equal to the algebraic degree of the highest-weight monomial in its ANF.

For example, the algebraic degree of $1 \oplus x_{2} \oplus x_{1} \oplus x_{3}x_{4} \oplus x_{1}x_{4} \oplus x_{1}x_{3} \oplus x_{1}x_{3}x_{4} \oplus x_{1}x_{2}x_{4}$ is 3.
\end{definition}

There exist various generalisations of the concept of algebraic degree to the case of $n \times m$ vectorial Boolean functions. The most common such definition \cite{Carlet2010_VBC} is as follows:

\begin{definition}\label{vectorial_degree_defn}
Consider the $2^{m}-1$ non-trivial Boolean functions $S_{i}$ $(i \in \{1, \ldots, 2^{m}-1\})$ that can be
obtained by forming linear combinations of the $m$ functions that map the input bits to individual
output bits of $S$. The algebraic degree of $S$ is defined as $\max_{i}(deg(S_{i}))$, i.e. the largest
algebraic degree of any of these functions $S_{i}$.
\end{definition}

\subsection{Transformations preserving relevant S-box properties, and notions of S-box equivalence.}
Various equivalence notions exist according to which there may be several S-boxes in the search space of bijections over $GF(2^{n})$ with identical differential uniformity and nonlinearity; and indeed with the same sets of absolute values in their DDTs and LATs \cite{Budaghyan_Carlet_Pott2006}. This has the potential to be particularly problematic for genetic and memetic algorithms, since it implies that many different ``genes'' may result in the same cryptanalytically relevant properties (see Appendix \ref{POWERS_Appendix} for more information on this.) We define these notions here:

\begin{definition}
Let $S_{1}$, $S_{2}$ be two S-boxes.

$S_{1}$ and $S_{2}$ are \textit{affine-equivalent} iff $S_{2} = A{\cdot}S_{1}{\cdot}B$, where $A, B$ are bijective affine transformations (so $A(x)$ would be the result of applying the transformation $M_{A}x \oplus V_{A}$, where $M_{A}$ was some invertible matrix and $V_{A}$ a vector.)
\end{definition}

\begin{definition}\label{EA_equivalence_defn}
Let $S_1$, $S_2$ be two S-boxes.

$S_1$ and $S_2$ are \textit{extended affine-equivalent} (EA-equivalent) iff $S_2 = A{\cdot}S_{1}{\cdot}B \oplus C$, where $A, B$ are bijective affine transformations and C is some, not necessarily bijective, affine transformation acting on the same input $x$ as $A{\cdot}S_{1}{\cdot}B$.
\end{definition}

\begin{definition}\label{CCZ_equivalence_defn}
Let $S_1$, $S_2$ be two S-boxes.

Let $gr(S_1)$ denote the graph of $S_1$, i.e. the set of all $(x, S_1(x))$ pairs. Let each such pair be viewed as a value in $GF(2)^{2n}$.

$S_1$ and $S_2$ are \textit{Carlet-Charpin-Zinoviev equivalent} (CCZ-equivalent for short) iff there exists some affine permutation $L: GF(2)^{2n} \rightarrow GF(2)^{2n}$ such that $L(gr(S_1)) = gr(S_2)$.

(so $L(x, S_1(x)) = (L_{1}(x, S_1(x)), S_2(L_{1}(x, S_1(x))))$, where $L_{1}: GF(2)^{2n} \rightarrow GF(2)^{n}$ maps $L$'s input bits to its first $n$ output bits.)
\end{definition}

\begin{definition}
An \textit{affine-invariant} property is a property which, if possessed by an S-box $S$, is also possessed by all S-boxes affine-equivalent to $S$.
Similarly, an EA-invariant property is a property which, if possessed by an S-box $S$, is also possessed by all S-boxes EA-equivalent to $S$, and CCZ-invariant properties are defined in terms of CCZ-equivalent S-boxes in the same way.
\end{definition}

It may be seen from the above definition that any two S-boxes which are affine equivalent are also EA-equivalent. It follows from this that all EA-invariant properties are also affine-invariant. Furthermore, CCZ-equivalence generalises EA-equivalence \cite{Budaghyan_Carlet_Pott2006}, so all CCZ-invariant properties are also EA-invariant and hence affine-invariant.

Importantly, differential uniformity and nonlinearity are CCZ-invariant properties \cite{Budaghyan_Carlet_Pott2006}.

We sought to find ways of utilising the equivalence classes to reduce the size of the search space, and to reveal patterns in the truth tables of the evolved mappings which could be exploited by the ant colony and memetic algorithm experiments (as well as some experiments with genetic algorithms that were later superseded by the memetics). The following theorem resulted from this:

\begin{theorem}\label{POWERS}
Every bijective S-box $S$ is affine-equivalent to at least one bijective S-box $S_2$ such that $S_2$ maps all inputs with Hamming weight less than 2 to themselves, 3 to 3 or 5 (we can restrict this to 5 if the S-box has differential uniformity $\leq 4$), 5 to some value $\leq 11$ (this may be restricted to 6 or 10 if the S-box has differential uniformity 2), and all $2^{i}+1 (3 \leq i \leq (n-1))$ to some value $\leq 2^{i+2}-2i-1$.
\end{theorem}

The proof, which also describes the procedure to construct the equivalent S-box, is included in Appendix \ref{POWERS_Appendix}.

\section{Experiments with simulated annealing}
\subsection{A description of simulated annealing}
Simulated annealing (SA) is a local-search based metaheuristic, inspired by a technique used in metallurgy to eliminate defects in the crystalline structures in samples of metal. A ``search space'', in this case the set of bijections over $GF(2^{n})$, is defined, in which we search for an entity possessing good properties according to our criteria. Every entity within this space is a ``candidate solution''.

The space is defined not merely by specifying the set of objects within it, but also by the ``move function''. The move function defines a transformation mapping one element in the search space to another. In our experiments, this function chose two of the S-box output values at random and swapped them round. If we were evolving bitstrings, the move function might flip the value of one of the bits, or change a 0 to a 1 and a 1 to 0 simultaneously. A ``move'', in this context, is the replacement of the current solution candidate with one obtained by applying the move function to it once, and the $x$-neighbourhood, or $x$-move-neighbourhood, of a candidate solution $C$ is defined as the set of all candidate solutions that can be obtained by making at most $x$ moves.

In simulated annealing, some initial candidate solution, $S_{0}$, usually chosen at random, is input to the SA algorithm, along with the following parameters:

\begin{itemize}
\item A \textit{cost function} $C$. The cost function takes a candidate solution as input, and outputs a scalar value, the ``cost''. The better the properties of the candidate solution in terms of what we wish to achieve, the lower the cost should be.

The cost function should also define a ``smooth search landscape'', in that there should exist some low value upper-bounding the extent to which the cost can change when one move is made.
\item The initial value $T_{0}$ for the ``temperature''. The higher the temperature, the more likely the search algorithm is to accept a move which results in a candidate solution with higher cost than the current candidate (that is, to store said candidate solution as the ``current candidate''). The temperature drops over time, causing the algorithm to accept fewer non-improving moves and hence to shift away from exploration and towards optimisation. Towards the end of the search, it is extremely rare for the algorithm to accept a non-improving move, and its behaviour is very close to that of a hill-climbing algorithm.
\item In choosing the value of $T_{0}$, various sources state that it should be chosen so that a particular proportion of moves are accepted at temperature $T_{0}$. There is very little information or advice available as to what this proportion should be. In \cite{Kirkpatrick1984} it is stated that any temperature leading to an initial acceptance rate of 80\% or more will do; however our initial experiments indicated that this was far too high and we eventually settled on an initial acceptance rate of 0.5 instead of 0.8.

Having chosen the initial acceptance rate, the experimenter executes the annealing algorithm with various $T_0$ until a temperature is found that achieves a fraction close enough to this. We started with the temperature at 0.1, and repeatedly ran the algorithm, doubled the temperature, and re-ran the algorithm until an acceptance rate at least as high as that specified was obtained. Where $T_a$ was the temperature at which this had been achieved, and $T_b = T_a/2$, we then used a binary-search-like algorithm to obtain a temperature between $T_a$ and $T_b$ that would result in an acceptance rate $\approx 50\%$.
\item A value $\alpha$; the ``cooling factor'', determining how far the temperature decreases at each iteration of the algorithm.
\item An integer value: $MAX\_INNER\_LOOPS$, determing the number of moves that the local search algorithm can make at each temperature.
\item The stopping criterion must also be specified. We used a $MAX\_OUTER\_LOOPS$ value, indicating how many times the algorithm was to be allowed to reduce the temperature and continue searching before it stopped.
\item We also specified a $MAX\_FROZEN\_OUTER\_LOOPS$ parameter. If the algorithm had, at any stage, executed this many outer loops without accepting a single move, it would be considered extremely unlikely to do anything other than remain completely stationary from then on, and would be instructed to terminate early.
\end{itemize}

In pseudocode:

\begin{algorithm}
\caption{Pseudocode for simulated annealing algorithm}
\begin{algorithmic}
\State $S \gets S_{0}$
\State $bestsol \gets S_{0}$
\State $T \gets T_0$
\State $ZERO\_ACCEPT\_LOOPS \gets 0$
\For{$x \gets 0, MAX\_OUTER\_LOOPS-1$}
	\State $ACCEPTS\_IN\_THIS\_LOOP \gets false$
	\For{$y \gets 0, MAX\_INNER\_LOOPS-1$}
		\State Choose some $S_{n}$ in the 1-move neighbourhood of $S$.
		\State $cost\_diff \gets C(S_{n}) - C(S)$
		\If{$cost\_diff < 0$}
			\State $S \gets S_{n}$
			\State $ACCEPTS\_IN\_THIS\_LOOP \gets true$
			\If{$C(S_{n}) < C(bestsol)$}
				\State $bestsol \gets S_{n}$
			\EndIf
		\Else
			\State $u \gets Rnd(0, 1)$
			\If{$u < exp(-cost\_diff/T)$}
				\State $S \gets S_{n}$
				\State $ACCEPTS\_IN\_THIS\_LOOP \gets true$
			\EndIf
		\EndIf
	\EndFor

	\If{$ACCEPTS\_IN\_THIS\_LOOP = false$}
		\State $ZERO\_ACCEPT\_LOOPS \gets ZERO\_ACCEPT\_LOOPS + 1$
		\If{$ZERO\_ACCEPT\_LOOPS = MAX\_FROZEN\_OUTER\_LOOPS$}
			\State \Comment Algorithm terminates early.
			\State \textbf{return} $bestsol$
		\EndIf
	\EndIf

	\State $T \gets T \times {\alpha}$
\EndFor

\State \textbf{return} $bestsol$
\end{algorithmic}
\end{algorithm}

\subsection{Our experiments}
We adopted the approach of Clark et al. of annealing with one cost function, and then hill-climbing with another. We also based our annealing cost functions on theirs \cite{Clark_Jacob_Stepney2004}, in that for whichever table $T$ was relevant to the criterion of primary interest the cost function took the form:

\begin{equation*}
\sum_{i=0}^{2^{n}-1} \sum_{j=0}^{2^{m}-1} ||T_{i, j}| - X|^{r}
\end{equation*}

This family of cost functions aimed to ``flatten'' the contents of the table $T$ as much as possible, achieving a relatively uniform table in which few entries deviated significantly from the value $X$. For a difference distribution table in particular, $X$ equal to 0 or 1 would be an intuitive choice, since for the best possible APN S-boxes all values were 0 or 2 with mean 1, and since low values in this table were desirable. In focusing on all the values in the table, instead of the single most extreme value, the search was better equipped to optimise table values that were not the extremal values at that particular point in time.

As stated earlier, we focused on differential uniformity and nonlinearity, meaning that $T$ would be either the DDT or the LAT. We experimented with various values of $X$ and $r$ for both the LAT and DDT, and tried multiplying and adding the cost function output values for different tables to achieve a multiobjective optimisation.

Our hill-climbing cost functions were not based solely on whichever of the differential uniformity and nonlinearity was being targeted. Since each of these was defined in terms of some extremal value in its corresponding table, we adapted the cost function to subtract $k$/(the number of times this value appeared) - where $k$ was the largest value for that table that could be guaranteed to divide all entries in it. ($k$ was equal to 2 in both cases). This gave us as cost functions:
\begin{itemize}
\item $DU - 2/DF$ and
\item $max_{i, j}|LAT_{ij}| - 2/NF$
\end{itemize}

By minimising the respective frequencies, we hoped to minimise the number of ways in which the cryptanalyst could exploit the extreme table values - in particular, we hoped to improve the S-box's resistance against linear cryptanalysis with multiple approximations \cite{Biryukov_DeCannière_MQuisquater2004, Cho_Hermelin_Nyberg2009, ChoPRESENTv2009}. We also hoped that this would guide the search towards lower values for the $DU$ and $max_{i, j}|LAT_{ij}|$. Part of this was based on the fact that using only $DU$ or $max_{i, j}|LAT_{ij}|$ to define the cost function would have made it harder for the hill-climbing stage to find improved candidate solutions if there were none with lower cost within the 1-move neighbourhood of the current candidate - in fact, whether it could do so at all would be dependent on the precise hill-climbing method implemented. Another motivation was the fact that the 1-move neighbourhoods of the known APN S-boxes contained differentially-4-uniform S-boxes with extremely low differential frequencies.

\subsubsection{Refining the annealing cost functions for differential uniformity}
During experiments with the following cost functions:

\begin{equation*}
\sum_{i=0}^{2^{n}-1} \sum_{j=0}^{2^{m}-1} ||DDT_{i, j}| - X|^{2}
\end{equation*}

we observed that where the random number generator (the boost::mt19937 Mersenne Twister from the C++ Boost libraries \cite{Boost}) had been seeded with the same seed, experiments for different values of $X$ were producing the same results.

To work out why this was, we first note that there are no negative values in the DDT, so the above equation is equivalent to $\sum_{i=0}^{2^{n}-1} \sum_{j=0}^{2^{m}-1} |DDT_{i, j} - X|^{2}$. Since, for any $d$, $d^{2} = |d|^{2}$, we have:

\begin{IEEEeqnarray*}{l}
\sum_{i=0}^{2^{n}-1} \sum_{j=0}^{2^{m}-1} ||DDT_{i, j}| - X|^{2} \\
= \sum_{i=0}^{2^{n}-1} \sum_{j=0}^{2^{m}-1} |DDT_{i, j} - X|^{2} \\
= \sum_{i=0}^{2^{n}-1} \sum_{j=0}^{2^{m}-1} (DDT_{i, j} - X)^{2} \\
= \sum_{i=0}^{2^{n}-1} \sum_{j=0}^{2^{m}-1} (DDT_{i, j}^{2} - 2{\cdot}X{\cdot}DDT_{i, j} + X^{2}) \\
= \sum_{i=0}^{2^{n}-1} \sum_{j=0}^{2^{m}-1} DDT_{i, j}^{2} - 2X \sum_{i=0}^{2^{n}-1} \sum_{j=0}^{2^{m}-1} DDT_{i, j} + \sum_{i=0}^{2^{n}-1} \sum_{j=0}^{2^{m}-1} X^{2}
\end{IEEEeqnarray*}

Note that $\sum_{i=0}^{2^{n}-1} \sum_{j=0}^{2^{m}-1} X^{2}$ is a constant value, independent of any values present in the DDT. Moreover, $\sum_{i=0}^{2^{n}-1} \sum_{j=0}^{2^{m}-1} DDT_{i, j}$ was shown in Definition \ref{DDT_defn} to be equal to $2^{m+n}$, and hence $2X \sum \sum DDT_{i, j}$ is equal to $2^{m+n+1}X$; another constant.

We see that the only part of the above cost function which can change at all when a move is made is $\sum_{i=0}^{2^{n}-1} \sum_{j=0}^{2^{m}-1} DDT_{i, j}^{2}$, and this is not dependent on the value of $X$.

Early experiments to compare the effects of the various DDT cost functions used:

\begin{itemize}
\item $n=6$,
\item Cooling factor 0.99 (we would later reduce this to 0.97),
\item $MAX\_INNER\_LOOPS=10000$,
\item $MAX\_OUTER\_LOOPS=500$,
\item $MAX\_FROZEN\_OUTER\_LOOPS = 200$.
\end{itemize}

These parameters had already been observed to find S-boxes with optimal differential uniformity, and optimal nonlinearity in some cases, for $n=5$. We had not, however, managed to find S-boxes with the optimal properties for $n=6$, and we felt that more information might be obtained from the differential frequencies of those S-boxes which achieved near-optimal differential uniformity of 4. In all cases, the hill-climb cost function was $DU - 2/DF$.
\\
\begin{table}[h]
\centering
\begin{tabular}{|l|l|l|l|l|l|}
\hline
	& \multicolumn{5}{|c|}{$X$} \\
\hline
	& -2		& -1		& 0		& 1		& 2 \\
\hline
$r=2$	& (59, 224.44)	& (59, 224.44)	& (59, 224.44)	& (59, 224.44)	& (59, 224.44) \\
\hline
$r=3$	& (65, 224.68)	& (71, 223.1)	& (58, 225.72)	& (59, 222.71)	& (71, 219,97) \\
\hline
$r=4$	& (71, 223.69)	& (68, 223.5)	& (67, 223.75)	& (76, 229.3)	& (75, 232.54) \\
\hline
\multicolumn{2}{|c|}{$DU - 2/DF$ to anneal} & \multicolumn{4}{|c|}{(0, N/A)} \\
\hline
\multicolumn{2}{|c|}{$r=2,\ X=0$ optimised} & \multicolumn{4}{|c|}{(78, 217.12)} \\
\hline
\end{tabular}
\caption{(Percentage of $DU\ 4$, average $DF$ for $DU\ 4$) for $n=6$. No boxes with $DU\ 2$ were found.}
\label{tab:du4n6_table}
\end{table}

The best $DU$ we found in any of our experiments for $n=6$ was 4, even though S-boxes of this size with $DU\ 2$ are known to exist \cite{Dillon2009_2}. Exponent 2 resulted in fewer $DU\ 4$ S-boxes being found than the other exponents; however the average $DF$ of these was not significantly different. However, using exponent 2 and $X=0$ allowed us to code a highly optimised version of the ``lookahead'' function that evaluated the change in cost when a move was being considered (the most generic approach to the lookahead would be to make the move, update relevant data, call the cost function, then undo the preceding) - and the result of this was that the search concluded in just over one-twelfth the amount of time as the fastest of the other cost functions. Increasing $MAX\_INNER\_LOOPS$ proportionately, we obtained 78\% $DU\ 4$ with average $DF\ 217.12$. We therefore accepted $r=2,\ X=0$ as the best cost function to use in large scale attempts to obtain low differential uniformity through simulated annealing.

\subsubsection{Relating nonlinearity and differential cost functions.}
The formula $\sum_{i=0}^{2^{n}-1} \sum_{j=0}^{2^{m}-1} ||T_{i, j}| - X|^{r}$ bears a strong resemblance to the formula for the sample variance of the entries of $T$. Indeed, for the DDT, let $r=2$, and let $X=1$, since this is the sample mean of the entries in the DDT of a bijective S-box, then apart from the division by $2^{n+m}$, the formulae are identical. Prior to our discovering that the value of $X$ did not affect the behaviour of the search algorithm for the DDT with exponent 2, this led to our experimenting with the variance of the DDT as a cost function. During these experiments, the variance of the absolute LAT values, as well as the variance of the squared LAT values (due to Theorem \ref{hadamard_ddt_lat_theorem}) was also recorded, and we noticed that in every case the latter was equal to $2^{2n-4}{\times}$ the DDT variance. However, the ratio between the two variances was not a power of 2, or indeed an integer, for non-bijective S-boxes. Investigating further, the same relationship was observed to hold between the sum of squares of the DDT and the sum of fourth powers of the LAT for a bijective S-box. Again, however, this was not the case for a non-bijective S-box.

We therefore make the following conjecture: For bijective S-boxes, the value of DDT cost function $r=2, X=0$ is always a fixed multiple of the value of LAT cost function $r=4, X=0$ (the precise multiple being determined by the value of $n$.) Although the temperatures required to achieve the desired initial acceptance rate may differ for these two cost functions, neither should be any more or less effective than the other in simulated annealing - they are in practice equivalent.

\begin{table}[h]
\centering
\begin{tabular}{|l|l|l|l|l|l|l|l|}
\hline
    & \multicolumn{7}{|c|}{$X$} \\
\hline
    & -6		& -4		& -2		& 0		& 2		& 4		& 6\\
\hline
r=2 & (55, 59.13)	& (72, 61.46)	& (67, 59.78)	& (N/A, N/A)	& (73, 61.53)	& (70, 58.69)	& (70, 58.31)\\
\hline
r=3 & (99, 48.46)	& (100, 41.94)	& (100, 41.22)	& (100, 37.18)	& (95, 50.4)	& (56, 62.98)	& (59, 50.61)\\
\hline
r=4 & (100, 40.83)	& (100, 41.04)	& (100, 40.24)	& (100, 34.3)	& (95, 49.4)	& (58, 64.98)	& (59, 61.3)\\
\hline
r=5 & (100, 41.35)	& (100, 41.78)	& (100, 40.8)	& (100, 32.7)	& (96, 49.7)	& (73, 61.11)	& (62, 63.42)\\
\hline
r=6 & (100, 42.18)	& (100, 43.28)	& (100, 42.65)	& (100, 32.45)	& (98, 52.02)	& (71, 59.92)	& (64, 64.88)\\
\hline
\multicolumn{4}{|c|}{DDT $(r=2,\ X=0)$ with speed/power tradeoff} & \multicolumn{4}{|c|}{(100, 30.39)} \\
\hline
\multicolumn{4}{|c|}{$max_{i, j}|LAT_{ij}| - 2/NF$ used to anneal} & \multicolumn{4}{|c|}{(100, 59.53)} \\
\hline
\end{tabular}
\caption{(Percentage of $NL\ 22$, average $NF$ for $NL\ 22$) for $n=6$. No $NL\ 24$ were found.}
\label{nl22nf_table}
\end{table}

In the experiments shown for the $6 \times 6$ problem size in Table \ref{nl22nf_table}, $(r=4, X=0)$ is not the best-performing cost function. Nevertheless, using the DDT cost function with $(r=2, X=0)$, the fastest of the other cost functions was 554 times slower in completing a search than it was. Increasing $MAX\_INNER\_LOOPS$ accordingly, as may be seen in the second-last row of the table, results in the best average $NF$ for $NL\ 22$. More generally, using the DDT cost function for the annealing stage offers the following advantages:

\begin{enumerate}
\item It allows us to use a lower exponent than we would otherwise, reducing the extent to which higher exponents would slow down the exponentiation involved.
\item As stated before, the optimised lookahead algorithm for DDT exponent 2 and $X=0$ can be used.
\item Updating the difference distribution table for a candidate solution when a move is made requires $O(2^{n})$ time. Updating the linear approximation table requires $O(2^{n+m})$ time. Similarly, the initial calculation of these tables also differs in complexity by a factor of $2^{m}$.
\end{enumerate}

We therefore accepted this cost function as the best candidate for larger-scale searches both for S-boxes with high nonlinearity and with low differential uniformity. Two sets of experiments, each consisting of 100 runs, were carried out. These used $\alpha = 0.97$ and, respectively, 3,000,000 and 30,000,000 inner loops per value of temperature.

\begin{table}[h]
\centering
\begin{tabular}{|l|l|l|l|}
\hline
max. inner loops	& \% DU 2 after		& \% NL 12 after 	& \% (DU 2, NL 12)\\
			& DU/DF hill-climb	& NL/NF hill-climb 	& after dual hill-climb\\
\hline
3,000,000		& 11			& 11 			& 6\\
\hline
30,000,000		& 45			& 34 			& 35\\
\hline
\end{tabular}
\caption{Experiments for $n=5$. For this size, NL cannot exceed 12 \cite{Chabaud_Vaudenay1994}.}
\label{tab:nfivethreetentable}
\end{table}

\begin{table}[h]
\centering
\begin{tabular}{|l|l|l|l|}
\hline
Inner loops	& (\% DU 4, avg. DF	& (\% NL 22, avg. NF		& \% (DU 4, NL 22)\\
		& for DU 4) after	& for NL 22) after 		& after dual hill-climb\\
		& DU/DF hill-climb	& NL/NF hill-climb		&	\\
\hline
3,000,000	& (80, 204.25)		& (100, 29.83)			& 35 (avg. DF 204, \\
		&			&				& avg. NF 36.09)\\
\hline
30,000,000	& (86, 197.65)		& (100, 28.06)			& 36 (avg. DF 199.25, \\
		&			&				& avg. NF 31.78)\\
\hline
\end{tabular}
\caption{Experiments for $n=6$. For this size, (DU 2, NL 24) and (DU 4, NL 24) S-boxes are known to exist but were not found. It is not known if NL 26 boxes exist.}
\label{tab:nsixthreetentable}
\end{table}

For the $n=5$ problem size, the APN S-boxes found fell into two categories. The boxes with nonlinearity 12 have properties consistent with the Gold exponent S-boxes \cite{Cheon_Lee2004} and their inverses, whereas the APNs with nonlinearity 10 appear to be ``inverse-based'' S-boxes \cite{Cheon_Lee2004}. All S-boxes with nonlinearity 12 were APN, this being due to the properties of almost bent S-boxes \cite{Chabaud_Vaudenay1994}. This means that we have, for this problem size, obtained results matching the theoretical best-possible (and the best S-boxes obtained through mathematical construction), as well as improving on previous applications of simulated annealing to this problem (which only achieved nonlinearity 10 - see Table \ref{MA_NL_table}.)

The best S-box we achieved for $n=6$ in terms of differential uniformity had $DU\ 4$, $DF\ 159$, $NL\ 22$, and $NF\ 24$, occurring as part of the 30,000,000 inner loop experiment. Unfortunately, since a $6 \times 6$ APN S-box has been constructed \cite{Dillon2009_2}, this means that we have not managed to match the best existing constructions for this size.

Although $6 \times 6$ S-boxes with nonlinearity 24 have been constructed mathematically \cite{Nyberg1993} \cite{Dillon2009_2}, we were unable to obtain any such for any choice of cost function. The best we were able to achieve had nonlinearity 22. The lowest corresponding $NF$ we found was 22 (unfortunately in conjunction with $DU\ 6$).

For larger problem sizes, we were unable to compete with the best known existing results. We did not achieve nonlinearity higher than 48 for $n=7$, or higher than 104 for $n=8$ (the best-known values for these sizes are 56 and 112). The best $DU$ possible for $n=7$ is 2, we achieved only $DU\ 6$ with best $DF$ 11. Similarly, the best known $DU$ for $n=8$ is 4; the best we achieved was 6 with $DF$ 185.

\begin{table}[h]
\centering
\begin{tabular}{|l|l|l|l|}
\hline
	& Highest known NL	& Previous best							& Highest evolved NL \\
	&			& evolved NL \cite{Clark_Jacob_Stepney2004}			& (this paper) \\
\hline
$n=5$	& 12			& 10								& 12 \\
\hline
$n=6$	& 24			& 22								& 22 \\
\hline
$n=7$	& 56			& 48								& 48 \\
\hline
$n=8$	& 112			& 104								& 104 \\
\hline
\end{tabular}
\caption{Comparison of best-known and evolved nonlinearities for $n=$ 5, 6, 7 and 8.}
\label{MA_NL_table}
\end{table}

\section{Experiments with memetic algorithms}

\subsection{A brief description of memetic algorithms}
Memetic algorithms \cite{Moscato1989} combine genetic algorithms with local optimisation, and have proven to be extremely effective search techniques.

There is some variation in their working - in particular, not every implementation for every problem domain will go through the four main stages in the same order that we do, and some will use more sophisticated machine learning techniques in the local optimisation stage instead of the straightforward hill-climbing we do. With this noted, we continue with our description.

The memetic algorithm maintains a ``population'' of candidate solutions, in the form of a multiset with size determined by the parameter $popsize$. Over several ``generations'' - analogous to the outer loops of simulated annealing - new populations $P_{i}$ will be derived from their immediate predecessors. Members of the population $P_{0}$ at the start of the algorithm are generated randomly and hill-climbed to local optima. In our impentation, an ``interim'' multiset contains the results of applying the various stages of the algorithm to $P_{i}$ - for the purposes of this paper we denote this set $PCP$ (Post-Crossover Population). $PCP$ is cleared at the start of each generation, and repopulated by the ``crossover'' operation in said generation. The members of $PCP$ are then altered during the ``mutation'' phase of the iteration, hill-climbed, and, during the ``selection'' phase, used to generate $P_{i+1}$.

The number of generations is one of the parameters - $NO\_OF\_GENERATIONS$, analogous to the $MAX\_OUTER\_LOOPS$ parameter of simulated annealing.

The crossover function, $cross(p_1, p_2)$, takes two ``parent'' solution candidates $p_1$ and $p_2$ as input, and outputs a ``child'' candidate solution $o_1$ which in some way combines features from both $p_1$ and $p_2$. Note that $o_1 = cross(p_1, p_2)$ is not necessarily equal to $o_2 = cross(p_2, p_1)$.

Several different crossover algorithms have been designed for the evolution of bijective functions (or, indeed, any entity representable as a permutation on a set of integers), and it is considered extremely important to choose a good crossover algorithm for the problem domain. We experimented with two different crossover methods; PMX (``Partially Mapped CROSSover'') and cycle crossover, both of which are described in detail in \cite{Mathias_McDaniel_Starkweather_cWhitley_dWhitley1991}. These two crossover methods were chosen because of their focus on the values of $x$ mapping to each output, instead of the order in which these outputs appeared.

\begin{itemize}
\item \textit{Cycle crossover} works as follows
\begin{align*}
\texttt{PARENT 1:}				& \mathtt{a\ b\ c\ \mathbf{d}\ e\ f\ g\ h\ i\ j}	\\
\texttt{PARENT 2:}				& \mathtt{c\ f\ a\ \mathbf{j}\ h\ d\ i\ g\ b\ e}
\end{align*}
\begin{equation*}
\texttt{(Randomly chosen cycle start point is marked in bold.)}
\end{equation*}

The element of Parent 1 at the cycle start point is copied into the child in the same position:
\begin{align*}
\texttt{CHILD:}					& \mathtt{?\ ?\ ?\ d\ ?\ ?\ ?\ ?\ ?\ ?}
\end{align*}

The element in the same position as Parent 2 is the next to be copied into the child. However, it is copied in into the same position in which it occurs in Parent 1:

\begin{align*}
\texttt{CHILD:}					& \mathtt{?\ ?\ ?\ d\ ?\ ?\ ?\ ?\ ?\ j}
\end{align*}

This process continues until the process returns us to the original cycle start point - in other words, when a ``loop'' or ``cycle'' has been created. In this case:

$(d, j) \to (j, e) \to (e, h) \to (h, g) \to (g, i) \to (i, b) \to (b, f) \to (f, d) \to (d, j)$ again.

\begin{align*}
\texttt{CHILD:}					& \mathtt{?\ b\ ?\ d\ e\ f\ g\ h\ i\ j}
\end{align*}

Any still-vacant positions in the child are then filled by copying in the corresponding values from Parent 2:

\begin{align*}
\texttt{CHILD:}					& \mathtt{c\ b\ a\ d\ e\ f\ g\ h\ i\ j}
\end{align*}

\item \textit{PMX crossover} begins by choosing two ``crossing points'' at random, as illustrated by the vertical lines in the below. The elements of Parent 1 between these points are copied into the child:
\begin{align*}
\texttt{PARENT 1:}				& \mathtt{a\ b\ |\ c\ d\ e\ f\ |\ g\ h\ i\ j}	\\
\texttt{PARENT 2:}				& \mathtt{c\ f\ |\ a\ j\ h\ d\ |\ i\ g\ b\ e}	\\
\texttt{CHILD:}					& \mathtt{?\ ?\ |\ c\ d\ e\ f\ |\ ?\ ?\ ?\ ?}
\end{align*}

Next, any elements of Parent 2 which have not already been copied into the child are copied in:
\begin{align*}
\texttt{PARENT 1:}				& \mathtt{a\ b\ |\ c\ d\ e\ f\ |\ g\ h\ i\ j}	\\
\texttt{PARENT 2:}				& \mathtt{c\ f\ |\ a\ j\ h\ d\ |\ i\ g\ b\ e}	\\
\texttt{CHILD:}					& \mathtt{?\ ?\ |\ c\ d\ e\ f\ |\ i\ g\ b\ ?}
\end{align*}

C was already copied in from Parent 1, so the element in position [0] cannot be equal to C. We see that the Parent 2 element in the same position as C in Parent 1 is A, and copy that into position [0]. Similarly, the final element cannot equal E, and so we put H in that position, the Parent 2 element in the same position as the E of Parent 1.

\begin{align*}
\texttt{PARENT 1:}				& \mathtt{a\ b\ |\ c\ d\ e\ f\ |\ g\ h\ i\ j}	\\
\texttt{PARENT 2:}				& \mathtt{c\ f\ |\ a\ j\ h\ d\ |\ i\ g\ b\ e}	\\
\texttt{CHILD:}					& \mathtt{a\ ?\ |\ c\ d\ e\ f\ |\ i\ g\ b\ h}
\end{align*}

The final unallocated position is trickier. We cannot copy F in, as it is already present in the child. We look for F in Parent 1, and find D in the corresponding position of Parent 2. Unfortunately, D has also been copied into the child by this point! We go on to look for D in Parent 1, and find that J is present in the same position of Parent 2 and has not been copied into the child, allowing us to complete the process:

\begin{align*}
\texttt{PARENT 1:}				& \mathtt{a\ b\ |\ c\ d\ e\ f\ |\ g\ h\ i\ j}	\\
\texttt{PARENT 2:}				& \mathtt{c\ f\ |\ a\ j\ h\ d\ |\ i\ g\ b\ e}	\\
\texttt{CHILD:}					& \mathtt{a\ j\ |\ c\ d\ e\ f\ |\ i\ g\ b\ h}
\end{align*}

\end{itemize}

Whichever crossover method we choose, the following two parameters are involved:
\begin{itemize}
\item $no\_of\_children$: When $p_1$ and $p_2$ are selected from $P_i$, this determines whether, if the crossover function is applied, it will merely be used to add $o_1 = cross(p_1, p_2)$ to $PCP$, or whether $o_2 = cross(p_2, p_1)$ will also be calculated and added.

If the crossover function is not applied, this determines whether $p_1$ alone, or both $p_1$ and $p_2$, are added to $PCP$.
\item $crossover\_probability$: When $p_1$ and $p_2$ are selected from $P_i$ during the crossover phase, this determines the probability of the crossover function being applied - i.e whether $o_1$ (and $o_2$, depending on the previous parameter), or $p_1$ and perhaps $p_2$, are added to $PCP$ in this generation.
\end{itemize}

We also need a ``mutation function'', $mutate(c)$, taking a candidate solution from $PCP$ as input, making a small random change (the ``mutation'') of some form to it, and returning the result (which replaces the original in $PCP$). In this instance, the mutation function makes one move as defined by the same local search methodology used for the simulated annealing and hill-climbs - that is, it swaps two truth table elements. Two parameters are relevant to this:

\begin{itemize}
\item $max\_possible\_mutations$: During the ``mutation phase'' of the algorithm, this defines the maximum number of mutations that may be applied to any single candidate.
\item $mutation\_probability$: Each potential mutation (up to the number defined by the candidate above) occurs randomly with this probability, independent of the other potential mutations.
\end{itemize}

Mutation adds an aspect of exploration into the memetic search, enabling it to escape from local optima.

The third phase, hill-climbing, is relatively simple. Note, however, that this phase and the one after it utilise a ``fitness function''. This is similar to the cost function we used in the simulated annealing experiments, except that high values are returned for high quality solutions instead of low values. In this phase, the members of $PCP$ are all hill-climbed to local optima with respect to the fitness function.

Finally, we have the ``selection'' phase, which is itself divided into various subphases. The implementer may decide to sort the elements of $PCP$ by their fitness values for the sake of efficiency at the start of the selection phase, if so this sorting is the first subphase.

After the sorting is carried out (or not), the next subphase is the ``elitism'' subphase. If the parameter $elitism\_level$ has a nonzero value, the $elitism\_level$ members of $P_{i}$ with the highest fitness values are copied directly into $P_{i+1}$. If this results in a full population (which is not advisable!) the selection phase ends. If not, we need to use a selection method to keep choosing elements from $PCP$ to add to $P_{i+1}$.

Let $|PCP|$ be denoted $M$. The two selection methods we experimented with here were:
\begin{enumerate}
\item \textbf{Roulette-wheel selection:} This method requires the fitness function to output a value $\geq 0$ for all possible inputs. Let $\sum_{i=0}^{M}\ fitness(c_i)$ be denoted $Z$. Let the number of places remaining in the population be denoted $r$. Then we follow the procedure in the pseudocode for Algorithm \ref{alg:roulette}:

\begin{algorithm}[H]
\caption{Pseudocode for roulette-wheel selection}
\label{alg:roulette}
\begin{algorithmic}
\For{$i \gets 1, r$}
	\State One member of $PCP$ is selected at random with probability $fitness(c_i)/Z$ \\
	\Comment{All $r$ selections are independent and at random.} \\
	\Comment{A candidate may be selected more than once.}
	\State A copy of this member is added to $P_{i+1}$.
	\State The original member is placed back in $PCP$.
\EndFor
\end{algorithmic}
\end{algorithm}
\item \textbf{Rank selection:} For this selection method, the members of $PCP$ must be sorted by fitness. Indexing from 1 upward, $PCP[1]$ is the candidate with the lowest fitness; $PCP[M]$ the candidate with the highest.

As before, in each of $r$ independent trials, a candidate is selected from $PCP$. A copy of this candidate is placed in $P_{i+1}$, and the candidate is replaced in $PCP$. The difference between this and roulette wheel selection is the probability with which each candidate is chosen:

\begin{equation*}
P(PCP[i]) = \frac{2i}{M(M+1)}
\end{equation*}
\end{enumerate}

\begin{algorithm}
\caption{Pseudocode for memetic algorithm}
\begin{algorithmic}
\State \Comment{Stage 1: Crossover.}
\State Reset $PCP$ to an empty multiset.
\While{$size(PCP) < POST\_CROSSOVER\_POPULATION\_SIZE$}
	\State Choose $p_1$ and $p_2$ from $P_{i}$ uniformly at random.
	\If{$Rnd(0, 1) < crossover\_probability$}
		\State $o_1 \gets cross(p_1, p_2)$
		\State Add $o_1$ to $PCP$
		\If{$no\_of\_children = 2$}
			\State $o_2 \gets cross(p_2, p_1)$
			\State Add $o_2$ to $PCP$
		\EndIf
	\Else
		\State Add $p_1$ to $PCP$
		\If{$no\_of\_children = 2$}
			\State Add $p_2$ to $PCP$
		\EndIf
	\EndIf
\EndWhile

\Comment{Stage 2: Mutation}
\For{$i \gets 0, POST\_CROSSOVER\_POPULATION\_SIZE-1$}
	\For{$j \gets 0, max\_possible\_mutations-1$}
		\If{$Rnd(0, 1) < mutation\_probability$}
			\State Apply one move (as defined for local search) to $PCP[i]$
		\EndIf
	\EndFor
\EndFor

\Comment{Stage 3: Hill-climbing}
\For{$i \gets 0, POST\_CROSSOVER\_POPULATION\_SIZE-1$}
	\State Hill-climb $PCP[i]$ to a local optimum.
\EndFor

\Comment{Stage 4: Selection.}
\State Reset the population to the empty multiset.
\If{$elitism\_level$ is specified}
	\State copy $elitism\_level$ members of $P_{i}$ into $P_{i+1}$
\EndIf

\While{$size(population) < popsize$}
	\State use a selection function to choose the next member of $PCP$ to add to $P_{i+1}$.
\EndWhile
\end{algorithmic}
\end{algorithm}

\subsection{Our experiments}
Due to its success and efficiency in the simulated annealing experiments, we focused exclusively on the sum of squares in the DDT as the basis for the fitness function. In all the below experiments, we carried out 100 runs of the memetic algorithm.

The first set of experiments varied the crossover method and selection method, as well as the $max\_possible\_mutations$ and $mutation\_probability$ criteria. The population size was set to 400, crossover probability to 1, $elitism\_level$ to 1, and $no\_of\_children$ to 2:

\begin{table}[h]
\centering
\begin{tabular}{|l|l|l|l|l|l|l|l|l|l|l|}
\hline
	& \multicolumn{10}{|c|}{$max\_possible\_mutations\ (m)\ {\times}\ mutation\_probability$} \\
\hline
	& 0.1     & 0.2    & 0.3    & 0.4    & 0.5    & 0.6    & 0.7     & 0.8    & 0.9    & 1.0 \\
\hline
$m=1$	& (9, 8)  & (6, 4) & (5, 3) & (6, 2) & (5, 5) & (4, 4) & (12, 7) & (7, 2) & (2, 0) & (11, 6) \\
\hline
$m=2$	& (12, 8) & (7, 5) & (5, 3) & (8, 6) & (7, 6) & (5, 4) & (10, 5) & (9, 4) & (6, 5) & (8, 4) \\
\hline
\end{tabular}
\caption{(Percentage of $DU\ 2$, percentage of $(DU\ 2, NL\ 12)$) for $n=5$ with cycle crossover and roulette-wheel selection. $(DU\ 2, NL\ 12)$ is the theoretical optimum with regard to our search criteria.}
\label{cycle_roulette_table}
\end{table}

\begin{table}[h]
\centering
\begin{tabular}{|l|l|l|l|l|l|l|l|l|l|l|}
\hline
	& \multicolumn{10}{|c|}{$max\_possible\_mutations\ (m)\ {\times}\ mutation\_probability$} \\
\hline
	& 0.1    & 0.2    & 0.3    & 0.4    & 0.5    & 0.6    & 0.7    & 0.8    & 0.9    & 1.0 \\
\hline
$m=1$	& (2, 0) & (1, 1) & (1, 1) & (3, 3) & (3, 3) & (3, 0) & (4, 3) & (0, 0) & (3, 2) & (1, 1) \\
\hline
$m=2$	& (2, 1) & (0, 0) & (2, 1) & (2, 2) & (2, 2) & (2, 1) & (2, 1) & (2, 0) & (3, 2) & (6, 1) \\
\hline
\end{tabular}
\caption{(Percentage of $DU\ 2$, percentage of $(DU\ 2, NL\ 12)$) for $n=5$ with cycle crossover and rank selection.}
\label{cycle_rank_table}
\end{table}

\begin{table}[h]
\centering
\begin{tabular}{|l|l|l|l|l|l|l|l|l|l|l|}
\hline
	& \multicolumn{10}{|c|}{$max\_possible\_mutations\ (m)\ {\times}\ mutation\_probability$} \\
\hline
	& 0.1     & 0.2    & 0.3     & 0.4     & 0.5     & 0.6     & 0.7    & 0.8    & 0.9     & 1.0 \\
\hline
$m=1$	& (11, 6) & (6, 4) & (9, 7)  & (9, 4)  & (13, 8) & (11, 5) & (8, 4) & (8, 6) & (12, 7) & (6, 3) \\
\hline
$m=2$	& (4, 4)  & (6, 3) & (12, 6) & (11, 7) & (7, 4)  & (6, 4)  & (9, 7) & (9, 2) & (10, 5) & (3, 2) \\
\hline
\end{tabular}
\caption{(Percentage of $DU\ 2$, percentage of $(DU\ 2, NL\ 12)$) for $n=5$ with PMX crossover and roulette wheel selection.}
\label{PMX_roulette_table}
\end{table}

\begin{table}[h]
\centering
\begin{tabular}{|l|l|l|l|l|l|l|l|l|l|l|}
\hline
	& \multicolumn{10}{|c|}{$max\_possible\_mutations\ (m)\ {\times}\ mutation\_probability$} \\
\hline
	& 0.1      & 0.2    & 0.3    & 0.4     & 0.5     & 0.6      & 0.7     & 0.8    & 0.9      & 1.0 \\
\hline
$m=1$	& (14, 8)  & (6, 4) & (8, 8) & (9, 8)  & (9, 4)  & (16, 10) & (14, 5) & (9, 7) & (16, 10) & (12, 5) \\
\hline
$m=2$	& (15, 11) & (8, 5) & (9, 3) & (14, 7) & (11, 5) & (8, 6)   & (12, 9) & (5, 2) & (16, 8)  & (14, 8) \\
\hline
\end{tabular}
\caption{(Percentage of $DU\ 2$, percentage of $(DU\ 2, NL\ 12)$) for $n=5$ with PMX crossover and rank selection.}
\label{PMX_rank_table}
\end{table}

During earlier experiments with genetic algorithms (memetic algorithms without a hill-climbing stage), there had been reason to believe that, depending on the population size, a certain value of $max\_possible\_mutations\ {\times}\ mutation\_probability$ would prove to be optimal. In the above experiments, there is far too much variation among the results to draw any such conclusion. However, it is clear from Tables \ref{cycle_roulette_table}, \ref{cycle_rank_table}, \ref{PMX_roulette_table} and \ref{PMX_rank_table} that the combination of cycle crossover and rank selection performs much more poorly than the other three (crossover, selection) choices. Furthermore, the combination of PMX and rank selection has led to higher percentages (14, 15, 16) of APN S-boxes than any of the other combinations, so we opted to stick with this for the second set of experiments. Choosing $max\_possible\_mutations$ and $mutation\_probability$ was similarly difficult due to the extent to which the results varied - we eventually opted for 1 mutation with probability 0.6.

We compared the results of imposing restrictions based on Theorem \ref{POWERS} with the results of not doing so (as also of allowing the solution candidates to make moves during mutation/hill-climbing that would violate these constraints and then retransform.) However, for every set of parameters for which this was tried, it resulted in worsened performance. We did not therefore make use of Theorem \ref{POWERS} in the experiments which followed.

Using these parameters, we experimented with varying the crossover probability. As part of this, we re-ran the original (1 mutation, probability 0.6, crossover probability 1.0) experiment. As may be seen from Table \ref{crossprob_table}, the performance of the memetic search drops markedly as crossover probability is reduced.

\begin{table}[h]
\centering
\begin{tabular}{|l|l|l|l|l|l|l|l|l|l|}
\hline
0.1	& 0.2	& 0.3	& 0.4	& 0.5	& 0.6	& 0.7	& 0.8	& 0.9	& 1.0 \\
\hline
0	& 0	& 0	& 0	& 0	& 0	& 2	& 3	& 8	& 15 \\
\hline
\end{tabular}
\caption{Percentage of $DU\ 2$ for various crossover probabilities. Again, $n=5$.}
\label{crossprob_table}
\end{table}

\begin{table}[H]
\centering
\begin{tabular}{|l|l|l|l|}
\hline
Population size	& $\%\ DU\ 2$	&	$\%\ (DU\ 2,\ NL\ 12)$	&	Time taken (d:h:m:s) \\
\hline
200		& 1		&	1			&	00:01:46:01 \\
\hline
400		& 15		&	7			&	00:03:44:21 \\
\hline
800		& 18		&	11			&	00:07:27:48 \\
\hline
1600		& 44		&	22			&	00:15:06:30 \\
\hline
3200		& 66		&	37			&	01:06:14:15 \\
\hline
\end{tabular}
\caption{Memetic algorithm results for various population sizes with $n=5$.}
\label{popsize_table}
\end{table}

For this reason, we kept this fixed to 1.0 for the third set of experiments, in which we varied the size of the population. It may be seen from the results of these experiments in Table \ref{popsize_table} that the quality of the solutions increased with the size of the population, although the time required to obtain these solutions also increased. Nevertheless, the experiment with population size 3200 outperformed ($66\%\ DU\ 2$ instead of $45\%$) the simulated annealing experiment with 30,000,000 inner loops in Table \ref{tab:nfivethreetentable}; and both of these experiments required roughly the same amount of time.

\section{Experiments with ant colony optimisation}
\subsection{A brief description of ant colony optimisation.}
The first ant colony optimization method was Ant System, described in \cite{Colorni_Dorigo_Maniezzo1996} as a metaheuristic that might be applied to the Travelling Salesman Problem (TSP). Later refinements produced the more effective Ant Colony System \cite{Dorigo_Gambardella1997} \cite{Luke2009Metaheuristics}, which took a more elitist approach and achieved superior results against the TSP.

Any problem to which ant algorithms can be applied must be possible to represent as a graph. For the S-box experiments, the graph nodes are the values of $x$, and the graph is directional - an edge from node $x$ to node $y$ signifies that $S(x)=y$. Furthermore, each edge carries with it a cost - and unlike the conventional TSP, the edge leading from $y$ to $x$ may not have the same cost as that from $x$ to $y$ (making our problem more akin to the Asymmetric TSP).

The problem should also allow a useful cost function to be devised such that, during the construction of each candidate solution, the cost starts at zero and is increased whenever a new component is added until the final cost is derived. In this case, basing the cost function on the DDT allowed us to do this; we could, for each $S(i)$ that was to be assigned a value, calculate for each $j$ which values in the DDT would be increased, and could thus calculate what the new cost would be if the resultant DDT values were input directly to the cost function. As with the memetic algorithms, we used the sum of squared DDT entries as a cost function; with each new truth table entry assigned a value, we could deduce which DDT entries would increase by 2 and how this would affect the sum of squares.

The value $d_{ij}$ denotes the amount by which the cost is increased if the edge from node $i$ to node $j$ is added, i.e. if $S(i)$ is assigned the value $j$. While for some problems (such as the TSP) $d_{ij}$ is constant, here it is affected by the truth table values that have already been assigned, and so must be recalculated every time we need to add a value for $S(i)$.

The following parameters are involved:

\begin{itemize}
\item The particular type of ant algorithm. In our experiments, we compared Ant System, Dorigo's original Ant Colony System, and the version of ACS defined in Sean Luke's ``Essentials of Metaheuristics''\cite{Luke2009Metaheuristics}. Other algorithms exist; for instance ``AntNet'' \cite{DiCaro_Dorigo1998, DiCaro_Dorigo_Gambardella1998}, a specialised variant designed for network routing problems.
\item $hillclimb\_trails$. This is a boolean value which determines whether or not local optimisation is used during trail-building. Early experiments indicated that setting this to false always led to worse results, so we fixed it at true in the experiments that follow.
\item $next\_index\_method$. After adding an edge from $i$ to $j$, this parameter determines which node the ant should try to add an edge leading from next. We experimented with ``cycle'', in which the next node is node $j$, and ``increment'', in which the next node is node $(i+1)$.
\item $\alpha$ and $\beta$ are floating-point values. The value of $\alpha$ determines the amount of influence pheromone levels have on edge selection, and the value of $\beta$ determines the influence of $d_{ij}$.

We used control values $\alpha=1$, $\beta=2$ in our experiments.
\item $e$ - the elitist pheromone update parameter (Used only in ACS versions of the global update stage.)
The values 0.05, 0.1, 0.2, 0.3 and 0.4 were tried, with 0.1 as control value as per \cite{Bianchi_Dorigo_Gambardella2002} and \cite{Dorigo_Gambardella1997}.
\item $\tau_{0}$ - the initial amount of pheromone on each edge. Following guidance in \cite{Bianchi_Dorigo_Gambardella2002}, and based on calculations of the optimal cost, we set this to $1.0/(2^{n}*((2^{n}-1)*2^{n})/2.0)$.
\item $no\_of\_ants$. The control value for the number of ants was 10 in accordance with the arguments in \cite{Dorigo_Gambardella1997}.
\item $Q$ - a scalar value by which the amount of pheromone deposited in the global update is multiplied. In the paper in which Ant System was originally described \cite{Colorni_Dorigo_Maniezzo1996}, after experiments with $Q=$ 1, 100, and 10000, 100 was accepted as the ``experimentally determined optimal value''. However, in the descriptions of ACS in \cite{Dorigo_Gambardella1997, Dorigo_Stützle2000, Luke2009Metaheuristics}, $Q=1$ was implicitly used, and no other values were mentioned. We experimented with the values 1, 10, 100, and $(2^{n}*(2^{n}-1))/2\ =\ 496$ for $n=5$.
\item $q_{0}$ - for ACS algorithms, this dictates the probability that a given edge selection will use an elitist selection method instead of the exploratory Ant System method. The control value was 0.98, with 0, 0.1, 0.25, 0.5, and 0.75 also being tried. For Ant System, $q_{0}$ is always zero.
\item $\rho$ - the non-elitist pheromone update parameter (also known as the ``evaporation rate''). The control value was 0.1 in accordance with \cite{Bianchi_Dorigo_Gambardella2002} and \cite{Dorigo_Gambardella1997}. 0.05, 0.2, 0.3, 0.4, and 0.5 were also tried.
\end{itemize}

In the below pseudocode, which describes all three of Ant System, Dorigo's ACS, and Luke's ACS, $\tau_{ij}$ denotes the amount of pheromone on the edge corresponding to $S(i)=j$. An ant trail is deemed to be complete when every node has an edge leading from it; i.e. when every $i$ has been assigned an output value $j=S(i)$.

\begin{algorithm}
\caption{Pseudocode for ant algorithms}
\begin{algorithmic}
\State Set the amount of pheromone on each graph edge to $\tau_{0}$.
\State $best\_solution \gets$ some randomly generated solution candidate.

\For{$x \gets 0, no\_of\_iterations-1$}
	\State \Comment{Each ant builds a trail}
	\State Clear all ant trails
	\State (remove all edges, set current nodes of all ants to 0.)
	\While {ant trails incomplete}
		\For{$a \gets 0, no\_of\_ants-1$}
			\State Let $i_a$ denote ant $a$'s current node.
			\State Let the set of unassigned output values at this point be denoted $U$.
			\State \Comment Add an edge from $i_a$ to some node $j_a \in U$
			\State \Comment (based on the cost of the edge and level of
			\State \Comment pheromone on it).
			\State $q \gets rnd(0, 1)$
			\If{$q \leq q_{0}$}
				\State Choose node $j_a$ where $j_a$ is the value of $k$ corresponding
				\State to $max_{k\ \in\ U}(\tau_{i_{a}k}^{\alpha}{\cdot}d_{i_{a}k}^{\beta})$
			\Else
				\State Node $j_a$ is chosen from the set $U$ with probability:
				\State $\frac{(\tau_{i_{a}j_{a}}^{\alpha}{\cdot}d_{i_{a}j_{a}}^{\beta})}{\sum_{k\ \in\ U}(\tau_{i_{a}k}^{\alpha}{\cdot}d_{i_{a}k}^{\beta})}$
			\EndIf

			\If{ant method is Dorigo's original ACS}
				\State \Comment Decrease pheromone levels on chosen edge (local update).
				\State $\tau_{i_{a}j_{a}} \gets (1 - \rho){\cdot}\tau_{i_{a}j_{a}} + {\rho}{\cdot}\tau_{0}$
			\EndIf

			\Comment Update current node:
			\If{$next\_index\_method = CYCLE$}
				\State $i_{a} \gets j_{a}$
			\ElsIf{$next\_index\_method = ITERATE$}
				\State $i_{a} \gets (i_{a}+1)$ modulo $no\_of\_nodes$
			\EndIf
		\EndFor

		\If{$hillclimb\_trails$}
			\State Hill-climb all constructed solutions represented by the ant trails
			\State to local optima.
		\EndIf
	\EndWhile

	\State Let $best\_iteration$ be the ant which constructed the best solution in this iteration.
	\State Let $best\_itera\_sol$ be that solution.

	\If{$cost(best\_itera\_sol) < cost(best\_solution)$}
		\State $best\_solution \gets best\_itera\_sol$
	\EndIf

	\State Update pheromone levels on all edges (global update).
	\State (The method varies depending on the choice of ant algorithm).
\EndFor
\State \textbf{return} $best\_solution$
\end{algorithmic}
\end{algorithm}

\subsection{Our experiments}
As with the memetic algorithms, we used the sum of DDT squares (and the cumulative effect of each added edge on it) as a cost function. The first major set of experiments varied $Q,\ q_0$, and the ant method:

\begin{table}[h]
\centering
\begin{tabular}{|l|l|l|l|l|}
\hline
$Q$				& 1	& 10	& 100	& 496 \\
\hline
$(\%\ DU\ 2)$ (cycle)		& 9	& 5	& 9	& 6 \\
\hline
$(\%\ DU\ 2)$ (increment)	& 11	& 6	& 7	& 8 \\
\hline
\end{tabular}
\caption{$(\%\ DU\ 2)$ for $n=5$ with Ant System}
\label{antsystem_q_table}
\end{table}

\begin{table}[H]
\centering
\begin{tabular}{|l|l|l|l|l|}
\hline
         	& \multicolumn{4}{|c|}{$Q$} \\
\hline
         	& 1	& 10	& 100	& 496 \\
\hline
$q_{0}=0$	& 7	& 5	& 12	& 2 \\
\hline
$q_{0}=0.1$	& 3	& 3	& 6	& 8 \\
\hline
$q_{0}=0.25$	& 9	& 4	& 7	& 3 \\
\hline
$q_{0}=0.5$	& 6	& 7	& 5	& 2 \\
\hline
$q_{0}=0.75$	& 1	& 1	& 4	& 3 \\
\hline
$q_{0}=0.98$	& 1	& 1	& 4	& 3 \\
\hline
\end{tabular}
\caption{$(\%\ DU\ 2)$ for $n=5$ with various values of $Q$ and $q_{0}$ for Dorigo ACS (cycle).}
\label{Dorigo_cycle_q_qzero_table}
\end{table}

\begin{table}[H]
\centering
\begin{tabular}{|l|l|l|l|l|}
\hline
         	& \multicolumn{4}{|c|}{$Q$} \\
\hline
         	& 1	& 10	& 100	& 496 \\
\hline
$q_{0}=0$	& 9	& 1	& 2	& 5 \\
\hline
$q_{0}=0.1$	& 8	& 5	& 3	& 2 \\
\hline
$q_{0}=0.25$	& 10	& 3	& 3	& 1 \\
\hline
$q_{0}=0.5$	& 2	& 2	& 2	& 1 \\
\hline
$q_{0}=0.75$	& 6	& 1	& 6	& 0 \\
\hline
$q_{0}=0.98$	& 1	& 0	& 2	& 0 \\
\hline
\end{tabular}
\caption{$(\%\ DU\ 2)$ for $n=5$ with various values of $Q$ and $q_{0}$ for Dorigo ACS (increment).}
\label{Dorigo_increment_q_qzero_table}
\end{table}

\begin{table}[H]
\centering
\begin{tabular}{|l|l|l|l|l|}
\hline
         	& \multicolumn{4}{|c|}{$Q$} \\
\hline
         	& 1	& 10	& 100	& 496 \\
\hline
$q_{0}=0$	& 7	& 10	& 5	& 6 \\
\hline
$q_{0}=0.1$	& 4	& 7	& 8	& 5 \\
\hline
$q_{0}=0.25$	& 4	& 4	& 7	& 2 \\
\hline
$q_{0}=0.5$	& 5	& 10	& 8	& 7 \\
\hline
$q_{0}=0.75$	& 13	& 5	& 9	& 4 \\
\hline
$q_{0}=0.98$	& 5	& 5	& 3	& 5 \\
\hline
\end{tabular}
\caption{$(\%\ DU\ 2)$ for $n=5$ with various values of $Q$ and $q_{0}$ for Luke ACS (cycle)}
\label{Luke_cycle_q_qzero_table}
\end{table}

\begin{table}[H]
\centering
\begin{tabular}{|l|l|l|l|l|}
\hline
         	& \multicolumn{4}{|c|}{$Q$} \\
\hline
         	& 1	& 10	& 100	& 496 \\
\hline
$q_{0}=0$	& 5	& 11	& 6	& 10 \\
\hline
$q_{0}=0.1$	& 7	& 1	& 3	& 3 \\
\hline
$q_{0}=0.25$	& 7	& 5	& 8	& 12 \\
\hline
$q_{0}=0.5$	& 12	& 3	& 4	& 10 \\
\hline
$q_{0}=0.75$	& 8	& 5	& 6	& 7 \\
\hline
$q_{0}=0.98$	& 7	& 10	& 11	& 4 \\
\hline
\end{tabular}
\caption{$(\%\ DU\ 2)$ for $n=5$ with various values of $Q$ and $q_{0}$ for Luke ACS (increment).}
\label{Luke_increment_q_qzero_table}
\end{table}

In the above tables, we show only the percentage of $DU\ 2$, since for this size $NL\ 12$ occurs only for $DU\ 2$ and the cost function has no effect on whether the $DU\ 2$ S-box is one of the $NL\ 10$ or $NL\ 12$ boxes we have seen so far.

As before, the amount of variation in the results makes them hard to interpret. However, large $Q$ for Dorigo's ACS appears to be detrimental to its performance, and we set $Q$ to 1 for the second set of experiments. It was difficult to draw any similar conclusion for Luke's ACS, but we set $Q$ to 1 for future experiments with this for the sake of comparison. Likewise, too high a value of $q_{0}$ seemed detrimental to the performance of Dorigo's algorithm, so we set this to 0.

In the second set of experiments, we varied the evaporation rate $\rho$ and elitist update parameter $e$:

\begin{table}[H]
\centering
\begin{tabular}{|l|l|l|l|l|l|l|l|}
\hline
$\rho$		& 0.05	& 0.1	& 0.2	& 0.3	& 0.4	& 0.5	& Total \\
\hline
$(\%\ DU\ 2)$	& 8	& 8	& 7	& 4	& 5	& 7	& 39 \\
\hline
\end{tabular}
\caption{$(\%\ DU\ 2)$ for Ant System (increment) with various evaporation rates.}
\label{Ant_evap_Increment_Table}
\end{table}

\begin{table}[H]
\centering
\begin{tabular}{|l|l|l|l|l|l|l|l|}
\hline
$\rho$		& 0.05	& 0.1	& 0.2	& 0.3	& 0.4	& 0.5	& Total \\
\hline
$(\%\ DU\ 2)$	& 11	& 4	& 8	& 11	& 9	& 7	& 50\\
\hline
\end{tabular}
\caption{$(\%\ DU\ 2)$ for Ant System (cycle) with various evaporation rates.}
\label{Ant_evap_Cycle_Table}
\end{table}

\begin{table}[H]
\centering
\begin{tabular}{|l|l|l|l|l|l|l|l|}
\hline
& \multicolumn{7}{|c|}{Evaporation rate ($\rho$)} \\
\hline
		& 0.05	& 0.1	& 0.2	& 0.3	& 0.4	& 0.5	& Total \\
\hline
$e = 0.05$	& 3	& 6	& 8	& 4	& 11	& 8	& 40 \\
\hline
$e = 0.1$	& 5	& 8	& 10	& 8	& 6	& 4	& 41 \\
\hline
$e = 0.2$	& 3	& 6	& 10	& 7	& 6	& 7	& 39 \\
\hline
$e = 0.3$	& 9	& 7	& 11	& 6	& 9	& 6	& 48 \\
\hline
$e = 0.4$	& 7	& 6	& 9	& 4	& 6	& 8	& 40 \\
\hline
$e = 0.5$	& 6	& 4	& 9	& 7	& 8	& 6	& 40 \\
\hline
Total		& 33	& 37	& 57	& 36	& 46	& 39	& 248 \\
\hline
\end{tabular}
\caption{$(\%\ DU\ 2)$ for Luke ACS (cycle) with varying $e$ and ${\rho}$.}
\label{Luke_evap_epu_cycle_table}
\end{table}

\begin{table}[H]
\centering
\begin{tabular}{|l|l|l|l|l|l|l|l|}
\hline
& \multicolumn{7}{|c|}{Evaporation rate ($\rho$)} \\
\hline
		& 0.05	& 0.1	& 0.2	& 0.3	& 0.4	& 0.5	& Total \\
\hline
$e = 0.05$	& 4	& 18	& 4	& 10	& 7	& 7	& 50 \\
\hline
$e = 0.1$	& 5	& 9	& 9	& 4	& 10	& 7	& 44 \\
\hline
$e = 0.2$	& 8	& 10	& 5	& 9	& 6	& 5	& 43 \\
\hline
$e = 0.3$	& 4	& 9	& 4	& 5	& 3	& 10	& 35 \\
\hline
$e = 0.4$	& 8	& 8	& 7	& 3	& 5	& 5	& 36 \\
\hline
$e = 0.5$	& 5	& 7	& 11	& 7	& 6	& 6	& 42 \\
\hline
Total		& 34	& 61	& 40	& 38	& 37	& 40	& 250 \\
\hline
\end{tabular}
\caption{$(\%\ DU\ 2)$ for Luke ACS (increment) with varying $e$ and ${\rho}$.}
\label{Luke_evap_epu_tables}
\end{table}

\begin{table}[H]
\centering
\begin{tabular}{|l|l|l|l|l|l|l|l|}
\hline
& \multicolumn{7}{|c|}{Evaporation rate ($\rho$)} \\
\hline
		& 0.05	& 0.1	& 0.2	& 0.3	& 0.4	& 0.5	& Total \\
\hline
$e = 0.05$	& 9	& 6	& 8	& 5	& 8	& 6	& 42 \\
\hline
$e = 0.1$	& 9	& 10	& 6	& 3	& 7	& 8	& 43 \\
\hline
$e = 0.2$	& 3	& 5	& 15	& 8	& 6	& 5	& 42 \\
\hline
$e = 0.3$	& 7	& 5	& 9	& 7	& 6	& 6	& 40 \\
\hline
$e = 0.4$	& 3	& 5	& 8	& 5	& 6	& 7	& 34 \\
\hline
$e = 0.5$	& 3	& 4	& 7	& 7	& 7	& 4	& 32 \\
\hline
Total		& 34	& 35	& 53	& 35	& 40	& 36	& 233 \\
\hline
\end{tabular}
\caption{$(\%\ DU\ 2)$ for Dorigo ACS (cycle) with varying $e$ and $\rho$.}
\label{Dorigo_evap_epu_cycle_table}
\end{table}

\begin{table}[H]
\centering
\begin{tabular}{|l|l|l|l|l|l|l|l|}
\hline
& \multicolumn{7}{|c|}{Evaporation rate ($\rho$)} \\
\hline
		& 0.05	& 0.1	& 0.2	& 0.3	& 0.4	& 0.5	& Total \\
\hline
$e = 0.05$	& 7	& 5	& 3	& 6	& 8	& 6 	& 35 \\
\hline
$e = 0.1$	& 7	& 8	& 7	& 9	& 6	& 7	& 44 \\
\hline
$e = 0.2$	& 7	& 4	& 10	& 4	& 9	& 6	& 40 \\
\hline
$e = 0.3$	& 2	& 6	& 7	& 6	& 10	& 4	& 35 \\
\hline
$e = 0.4$	& 4	& 4	& 3	& 6	& 7	& 6	& 30 \\
\hline
$e = 0.5$	& 6	& 2	& 7	& 7	& 6	& 5	& 33 \\
\hline
Total		& 33	& 29	& 37	& 38	& 46	& 34	& 217 \\					
\hline
\end{tabular}
\caption{$(\%\ DU\ 2)$ for Dorigo ACS (increment) with varying $e$ and $\rho$.}
\label{Dorigo_evap_epu_incr_table}
\end{table}

The above tables show a similar level of variation in the achieved results to their predecessors, with no apparent patterns. In an attempt to obtain some information on the merits of the various techniques and parameter choices, we have summed the numbers of APN S-boxes achieved across all experiments for each algorithm. Again, these results have a great deal of variance, but it does appear that Dorigo ACS with increment index is underperforming compared to the other ACS variants. In fact, we believe that the above tables are evidence that Luke's ACS outperforms Dorigo's for this particular problem.

While this is not so certain, we also note that high values of $e$ do appear to impair performance for Dorigo ACS; there is a clear drop in performance for cycle index; and also for increment index with the exception of $e=0.05$ (which we believe to be a statistical outlier). No such pattern is visible for the values of $\rho$ tried, however.

Cycle clearly outperforms increment index for the Dorigo ACS experiments, and while the effective sample size for the Ant System experiments is smaller, cycle still outperforms increment in these experiments by a clear margin. This is not the case for the Luke ACS experiments, however the number of APNs in the increment-index experiments would be reduced below the number for cycle-index by a slight margin if the statistical outlier for $e = 0.05,\ \rho=0.1$ were disregarded. We therefore believe that while this differs by ant algorithm, cycle-index is in general more effective than increment-index for this problem.

For our third set of experiments (in which the number of ants was varied), we decided to conduct the experiments for both Luke ACS and Ant System, since it was not possible to draw a firm conclusion from the evidence so far as to which was the more effective. We set $\rho$ to 0.1 and (for Luke ACS) $e$ to 0.1. We also set $q_0$ to 0:

\begin{table}[H]
\centering
\begin{tabular}{|l|l|l|l|l|}
\hline
& \multicolumn{4}{|c|}{Number of ants} \\
\hline
		& 10	& 32	& 64	& 128 \\
\hline
Ant System	& 6	& 17	& 40	& 68 \\
\hline
ACS (Luke)	& 7	& 22	& 37	& 57 \\
\hline
\end{tabular}
\caption{$(\%\ DU\ 2)$ for varying numbers of ants.}
\label{No_of_ants_table}
\end{table}

The Ant System result, as well as outperforming the best run with the memetic algorithm, also did so in just over 19 hours as opposed to just over 34 for the memetic.

(As an aside, we had also experimented again with restricting S-box output values in accordance with Theorem \ref{POWERS}, believing that these restrictions would be necessary (though perhaps insufficient) for the search to succeed. Ant algorithms reward edges that feature in ``good'' solutions by increasing the amount of pheromone on them, but any given edge $(i, S(i)=j)$, could feature in any S-box with any set of properties - it would always be possible to make an affine transformation that ensured this was so - and we hoped that the truth table restrictions imposed in Theorem \ref{POWERS} would overcome this. Unfortunately, this did not turn out to be the case - as with the memetic algorithms, it resulted in poorer-quality solutions on average for all the ant algorithm variants and all values of $Q$/$q_0$ tried - and we had to abandon this line of inquiry.)

\section{Conclusions, and directions for future research.}
We have demonstrated that where the criteria are differential uniformity and nonlinearity - the two most important criteria for block cipher S-boxes - that metaheuristic search is capable of matching the best theoretical results for S-boxes of size $5 \times 5$ and smaller. Unfortunately, the significant increase in the size of the search space for higher $n$ means that the difficulty level of the problem increases extremely rapidly, regardless of the metaheuristic used, and we were not able to achieve the same success for any $n \geq 6$.

Experimenting with various cost functions, we have found a particularly fast and effective cost function for this particular problem. We have also experimented with various parameter choices and found particularly effective parameters for three different types of metaheuristic applied to this problem. In comparing the metaheuristics, we have observed that ant algorithms - and in particular Ant System - appear to be more effective than either memetic algorithms or simulated annealing. This is also the first time ant algorithms and memetic algorithms have been applied to the S-box problem.

It seems unlikely that evolutionary methods acting on the truth table alone will be sufficient to find almost perfect nonlinear (or $DU\ 4$ with low $DF$) S-boxes for $n=6$ or higher. In future research, it may prove beneficial to try alternate representations of the S-box \cite{Clark_Jacob_Maitra_Stanica2004}, however it is not currently known if any alternate representations with suitable search landscapes exist.

Further research focusing on the use of theoretical results to impose further constraints on the truth table values (and hence the search space) may also yield a breakthrough. Currently, the restrictions imposed by Theorem \ref{POWERS} act on a fraction of truth table entries that decreases exponentially with $n$, and leave a great deal of latitude for the remaining values.

Furthermore, Theorem \ref{POWERS} only focuses on affine equivalence due to the difficulty in finding EA and CCZ transformations that achieved the results desired while preserving bijectivity. Attempting to find such transformations may achieve the stronger restrictions required, especially given that replacing a bijective S-box with its inverse is a CCZ transformation that preserves bijectivity.

Another possibility for future research might be to look at ways to apply a ``divide-and-conquer'' methodology to the problem - for example, evolving two subfunctions with low differential uniformity that would then join up into the larger function. This would not guarantee that the main function had low differential uniformity, but since it would be necessary for, say, two $5 \times 6$ S-boxes to have differential uniformity 2 if their $6 \times 6$ concatenation was to, it might lead to a way of tackling a smaller search space to solve the problem, or produce S-boxes closer to the optimum.

Finally, the cost of a given edge in the ant colony experiments varied depending on the other edges added beforehand. This may indicate that a version of ant colony optimization designed for dynamic problems, such as AntNet, might be possible to adapt to obtain superior performance.

\bibliographystyle{plain_misccommafix}
\bibliography{unifiedbib}
\begin{appendices}
\appendix
\section{The proof of Theorem \ref{POWERS}}\label{POWERS_Appendix}
We show that given any bijection $S$ over $GF(2^{n})$, an affine-equivalent bijection $S_2$ exists such that $S_2$ maps 0 and all $2^{i}$ to themselves (i.e. all $(n+1)$ values with Hamming weight $< 2$ are fixed points), and all $2^{i}+1$ to values within a certain restricted range. We then show that for $S$ with differential uniformity 2 or 4, the range of output values for certain inputs can be made even narrower. We do not succeed, however, in obtaining a unique representation - or ``normal form'' for the
affine equivalence class of $S$ in this way; and the ramifications of this when trying to evolve bijections over $GF(2^{n})$ (such as cryptographic S-boxes) using memetic or ant-based metaheuristic algorithms are considered.

\subsection{Preliminaries.}\label{SectionPreliminariesOne}

\begin{lemma}
Let $S$ denote a bijective S-box. It is trivial to construct another $n \times n$ S-box, $S_2$, which is EA-equivalent to $S$ and which maps all inputs with Hamming weight $\leq 1$ to themselves. However, $S_{2}$ is not necessarily bijective.
\end{lemma}
\begin{proof}
Let $S_2 = S \oplus Cx \oplus d$. Let $d$ be the bitstring representation of $S(0)$. By choosing the matrix $C$ so that it will map every input $x$ with Hamming weight 1 to $S(x) \oplus S(0) \oplus x$, we obtain an $S_2$ with the properties described.

Such a $C$ can be constructed by letting the $i$th column be the bitstring representation of $2^{n-i} \oplus S(2^{n-i}) \oplus S(0)$.

We now present an example where $S_2$ is constructed as described above from some bijective $S$, but is not itself bijective. Let $S$ be the S-box from the Courtois Toy Cipher \cite{Courtois2006}:

\begin{equation*}
\begin{array}{cccccccc}
0 & 1 & 2 & 3 & 4 & 5 & 6 & 7 \\
7 & 6 & 0 & 4 & 2 & 5 & 1 & 3
\end{array}
\end{equation*}

The matrix $C$ must then be defined as follows:

\begin{equation*}
C=
\begin{bmatrix}
0 & 1 & 0\\
0 & 0 & 0\\
1 & 1 & 0
\end{bmatrix}
\end{equation*}

Note that $S(0) {\oplus} C{\cdot}0 = 7 {\oplus} 0 = 7$. However, $S(7) {\oplus} C{\cdot}7 = 3 {\oplus} 4 = 7$.

It follows that $S_2(0) = S_2(7) = 7 \oplus d$. (Since $d$ is defined as $S(0)$, $d=7$.)
\end{proof}

In the following section, we shall prove that at least one S-box affine-equivalent to $S$, and mapping all inputs with Hamming weight $\leq 1$ to themselves as described, but which is also bijective, must always exist. We shall describe a procedure to construct this S-box from $S$; firstly by constructing an S-box mapping 0, 1 and 2 to themselves, and then applying a more complicated procedure to construct the final S-box from this. Although we do not manage to prove that any other values can be mapped to themselves without jeopardising affine equivalence, we do limit the range of values that can be mapped to by some of the S-box inputs with Hamming weight 2.

\subsection{Constructing the equivalent bijection.}\label{construction_main_section}

We shall construct a sequence of S-boxes, the first of which we shall construct from the original S-box, and will construct each successive S-box from its predecessor until we obtain one with the desired properties.

\begin{lemma}\label{SPR_two_point_one}
Let $S_{1}$ be a bijective S-box.

Then there exists at least one bijective S-box $S_2$ with $S_2(0)=0$, and which is also affine-equivalent to $S_1$.
\end{lemma}
\begin{proof}
To construct $S_2$, we can either
\begin{itemize}
\item Xor every output of $S_1$ with $S_1(0)$, or
\item Let $S_2$ be defined as $S_2(x) = S_1(x {\oplus} S_{1}^{-1}(0))$.
\end{itemize}
\end{proof}

\begin{lemma}
If $S_1$ has differential uniformity $< 2^n$, there exist at least two $S_2$ equivalent to $S_1$ with the properties described above.
\end{lemma}
\begin{proof}
If $S_{1}(0) = 0$, let $S_{a}$ be some affine-equivalent S-box such that this is not the case. Perhaps we could xor all of $S_{1}$'s outputs with 1 to achieve this. If $S_1(0) \neq 0$, let $S_a = S_1$.

Let $S_{2a}$ be the box we obtain by xoring every output of $S_{a}$ with $S_{a}(0)$.

Let $S_{2b}$ be the box defined by $S_{2b}(x) = S_{a}(x {\oplus} S_{a}^{-1}(0))$.

(Note that one of these may be the same S-box as the original $S_{1}$.)

If $S_{2a}$ and $S_{2b}$ were the same S-box, it would follow that $S_{a}(x {\oplus} S_{a}^{-1}(0))$ was equal to $S_{a}(x) {\oplus} S_{a}(0)$ in all cases. The entry in $S_{a}'s$ difference distribution table for row $S_{a}^{-1}(0)$ and column $S_{a}(0)$ would then be equal to $2^n$. However, since differential uniformity is affine-invariant, this would lead to a contradiction, as $S_1$ does not have differential uniformity $2^n$.
\end{proof}

If $n$ is equal to 1, then this is enough to prove the main result, since the bijectivity of the constructed S-box would mean that it also had to map 1 to itself. We shall therefore assume from here on that $n \geq 2$.

\begin{lemma}
For any nonzero $n$-vector, there exists at least one linear bijection expressed as an $n \times n$ matrix over $GF(2)$ such that said vector is a column thereof.

Furthermore, for any distinct pair of nonzero $n$-vectors, there exists at least one linear bijection expressed as an $n \times n$ matrix over $GF(2)$ such that each of these vectors is a column thereof.
\end{lemma}
\begin{proof}
The result is trivial. The only condition imposed on the matrix by its bijectivity is that it must be invertible, which is the case if and only if all its columns are linearly independent. This does not prevent us from choosing the first column arbitrarily (as long as it is nonzero), and the only restriction imposed on the second column is that it should not be equal to zero or to the first column. We choose the rest of the columns accordingly, and then reorder them if we do not wish the vector (or vectors) we started with to occupy the first column (or first two columns.)
\end{proof}

\begin{corollary}\label{corollary_4}
For all nonzero $x \in GF(2^n)$, and for any $0 \leq k < n$, there exists at least one bijective matrix $M$ such that $M(2^k) = x$.
\end{corollary}
\begin{proof}
Simply choose $M$ as described above so that its $(n-k)$th column is the vector $x$.
\end{proof}

\begin{corollary}\label{corollary_4_1}
For all distinct nonzero $(x, y) \in GF(2^n)$, and for any $0 \leq k < l < n$, there exists at least one bijective matrix $M$ such that $M(2^k) = x$ and $M(2^l) = y$.
\end{corollary}
\begin{proof}
Choose M so that its $(n-k)$th column is $x$ and its $(n-l)$th column is $y$.

Note that $M(2^{k} \oplus 2^{l})$ will be equal to $(x {\oplus} y)$.
\end{proof}

\begin{corollary}\label{corollary_5}
For all nonzero $x \in GF(2^n)$, and for any $0 \leq k < n$, there exists at least one bijective matrix $M$ such that $M(x) = 2^{k}$.
\end{corollary}
\begin{proof}
Choose some $N$ as described in Corollary \ref{corollary_4} such that $N(2^k)=x$. Then $N^{-1}$ is $M$ as desired.
\end{proof}

\begin{corollary}\label{corollary_5_1}
For all distinct nonzero $(x, y) \in GF(2^n)$, and for any $0 \leq k < l < n$, there exists at least one bijective matrix $M$ such that $M(x) = 2^k$ and $M(y) = 2^l$.
\end{corollary}
\begin{proof}
Choose some $N$ as described in Corollary \ref{corollary_4_1} such that $N(2^k)=x$ and $N(2^l)=y$. Then $N^{-1}$ is $M$ as desired.

Note that since $N(2^{k} {\oplus} 2^{l})$ will equal $x {\oplus} y$, $M(x {\oplus} y) = (2^{k} \oplus 2^l)$.
\end{proof}

\begin{theorem}\label{theorem_7_1}
Let $S$ be a bijective S-box such that $S(0)=0$. For any $0 \leq k < l < n$, there exists at least one affine-equivalent bijective S-box $S_{2}$ such that:
\begin{itemize}
\item $S_{2}(0)=0$,
\item $S_{2}(2^k) = 2^k$,
\item $S_{2}(2^l) = 2^l$, and
\item $S_{2}(x) = M{\cdot}S(x)$ for some linear bijective matrix $M$.
\end{itemize}
\end{theorem}
\begin{proof}
Any bijective linear transformation of 0 is 0, so $S_{2}(0) = 0$.

From Corollary \ref{corollary_5_1}, we can choose $M$ to be a transformation mapping $S(2^k)$ to $2^k$ and $S(2^l)$ to $2^l$.
(Again, we note that $M$ will also map $(S(2^k) \oplus S(2^l))$ to $(2^k \oplus 2^l)$.)
\end{proof}

\begin{theorem}\label{theorem_8_1}
Let $S$ be a bijective S-box such that $S(0)=0$. For any $0 \leq k < l < n$, there exists at least one affine-equivalent bijective S-box $S_{2}$ such that:
\begin{itemize}
\item $S_{2}(0)=0$,
\item $S_{2}(2^k) = 2^k$,
\item $S_{2}(2^l) = 2^l$, and
\item $S_{2}(x) = S{\cdot}M(x)$ for some linear bijective matrix $M$.
\end{itemize}
\end{theorem}
\begin{proof}
Any bijective linear transformation of 0 is 0, so $S_{2}(0) = 0$.

From Corollary \ref{corollary_4_1}, we can choose $M$ to be a transformation mapping $2^k$ to $S^{-1}(2^k)$ and $2^l$ to $S^{-1}(2^l)$. (Again, we note that $M$ will also map $(2^k {\oplus} 2^l)$ to $(S^{-1}(2^k) {\oplus} S^{-1}(2^l))$.)
\end{proof}

We see that it is fairly straightforward to construct an affine-equivalent bijective S-box mapping 0, 1, and 2 to themselves. If $n=2$, we have now achieved the desired result, so we shall now assume $n \geq 3$.

\begin{lemma}
Let $S$ be an APN S-box mapping 0, 1 and 2 to themselves. $S(3)$ cannot be equal to 3.
\end{lemma}
\begin{proof}
If $S(3)$ were equal to 3, (input difference 3, output difference 3) would occur for the input pairs (0, 3), (3, 0), (1, 2), and (2, 1), meaning that the differential uniformity of $S$ would be at least 4 and contradicting the statement that it is APN.
\end{proof}

We note that for values of $n$ such as 4, for which APN bijections do not exist \cite{Leander_Poschmann2007}, we cannot guarantee that a given S-box mapping 0, 1 and 2 to themselves will not also map 3 to itself, but it seems highly unlikely that boxes mapping 0, 1, 2 and 3 to themselves will not be in the minority. Certainly APN bijective S-boxes are known to exist for $n=6$ \cite{Dillon2009_2} and for all odd $n$ \cite{Nyberg1993}; furthermore differentially-4-uniform S-boxes are known to exist for all $n$ \cite{Nyberg1993} and a later result in this paper will show that, given some bijective differentially-4-uniform S-box mapping 0, 1, 2 and 3 to themselves, we can construct from it an affine-equivalent differentially-4-uniform bijective S-box mapping 0, 1, 2 and 3 to 0, 1, 2 and 5 respectively.

\begin{theorem}\label{s_three_main_body}
Let $S$ be a bijective S-box mapping 0, 1, 2 to themselves, but not mapping 3 to itself. There exist bijective S-boxes $S_1, S_2, S_3$ affine-equivalent to $S$ which also map 0, 1, and 2 to themselves, such that:
\begin{itemize}
\item $S_1(3) = 5$
\item $S_2(3) = 6$
\item $S_3(3) = 7$
\end{itemize}
\end{theorem}
\begin{proof}
If $S(3) \leq 7$, we can construct the desired $S_i$ immediately, by applying a process which we will refer to as "controlled-XOR":

\begin{definition}\label{cxor_definition}
Let $CXOR(i, a_{1}a_{2}{\ldots}a_{i-1})$ be a matrix identical to the identity matrix except in its $i$th-last column

Instead of

\begin{equation*}
\begin{vmatrix}
0 \\
\ldots \\
0 \\
1 \\
0 \\
0 \\
\ldots \\
0 \\
\end{vmatrix}
\end{equation*}

let it be

\begin{equation*}
\begin{vmatrix}
0 \\
\ldots \\
0 \\
1 \\
a_{1} \\
a_{2} \\
\ldots \\
a_{i-1} \\
\end{vmatrix}
\end{equation*}

Applying $CXOR$ to the outputs of the S-box will map all values $< 2^{i-1}$ to themselves, as also any values $> 2^{i-1}$ whose $i$th-last bit is zero.
$CXOR$ will map any output values with nonzero $i$th-last bit from

\begin{equation*}
\begin{vmatrix}
c_{1} \\
\ldots \\
c_{p} \\
1 \\
b_{1} \\
b_{2} \\
\ldots \\
b_{i-1} \\
\end{vmatrix}
\end{equation*}

to

\begin{equation*}
\begin{vmatrix}
c_{1} \\
\ldots \\
c_{p} \\
1 \\
b_{1} \oplus a_{1} \\
b_{2} \oplus a_{2} \\
\ldots \\
b_{i-1} \oplus a_{i-1} \\
\end{vmatrix}
\end{equation*}

and we simply choose $(a_1, {\ldots}, a_{i-1})$ to give the desired result.
\end{definition}

In this case, we simply apply $CXOR$ with i=3 and $a_{1}, a_{2}$ chosen to give whichever of 5, 6 and 7 is desired as the value for $S(3)$.

If $S(3) > 7$, we first create an S-box $S_{a}$ mapping 0, 1 and 2 to themselves and such that $S_{a}(3) \leq 7$.

This will involve applying a matrix which we will refer to as $MSB\_SHIFT$; which is similar to the matrices which comprise the $SWAP$ and $CNOT$ quantum gates.

\begin{definition}\label{MSBSHIFT_matrix_definition}

Let $MSB\_SHIFT(S(i))$ be a matrix identical in all but two rows to the identity matrix. These rows must be: the row in which the first (i.e. the MSB - the topmost when it is expressed as a column vector) nonzero bit of $S(i)$ occurs, and the row immediately below it. Let them be denoted $R1$ and $R2$ respectively.

Whereas the identity matrix includes:

\begin{align*}
\texttt{ROW ABOVE R1: }	& \mathtt{0 \ldots 1000 \ldots 0} \\
\texttt{R1: }		& \mathtt{0 \ldots 0100 \ldots 0} \\
\texttt{R2: }		& \mathtt{0 \ldots 0010 \ldots 0} \\
\texttt{ROW BELOW R2: }	& \mathtt{0 \ldots 0001 \ldots 0} \\
\end{align*}

the corresponding rows of $MSB\_SHIFT(S(i))$ are:

\begin{align*}
\texttt{ROW ABOVE R1: }	& \mathtt{0 \ldots 1000 \ldots 0} \\
\texttt{R1: }		& \mathtt{0 \ldots 0A10 \ldots 0} \\
\texttt{R2: }		& \mathtt{0 \ldots 01X0 \ldots 0} \\
\texttt{ROW BELOW R2: }	& \mathtt{0 \ldots 0001 \ldots 0} \\
\end{align*}

where:

\begin{itemize}
\item $A$ is the value of the bit immediately following the first nonzero bit of $S(i)$ (so $S(i)$ is of the form $1Abc {\ldots} {\omega}$)
\item $X$ is either 0 or $(1 \oplus A)$.
\end{itemize}

It is easy to confirm that the columns of $MSB\_SHIFT(S(i))$ are all linearly independent.

We transform $S$ by applying $MSB\_SHIFT(S(i))$ to the outputs of $S$; i.e. by calculating\\
$MSB\_SHIFT(S(i))(S)$. The effect of $MSB\_SHIFT(S(i))$ on $S(i)$ (indeed, on all outputs of $S$) is to map

\begin{equation*}
\begin{vmatrix}
0 \\
\ldots \\
0 \\
1 \\
A \\
b \\
\ldots \\
\omega \\
\end{vmatrix}
\end{equation*}

to

\begin{equation*}
\begin{vmatrix}
0 \\
\ldots \\
0 \\
(A{\cdot}1) \oplus (1{\cdot}A) = 0 \\
(1{\cdot}1) \oplus (X{\cdot}A) = 1 \\
b \\
\ldots \\
\omega \\
\end{vmatrix}
\end{equation*}

and

\begin{equation*}
\begin{vmatrix}
0 \\
\ldots \\
0 \\
0 \\
1 \\
b \\
\ldots \\
\omega \\
\end{vmatrix}
\end{equation*}

to

\begin{equation*}
\begin{vmatrix}
0 \\
\ldots \\
0 \\
(A{\cdot}0) \oplus (1{\cdot}1) = 1 \\
(1{\cdot}0) \oplus (X{\cdot}1) = X \\
b \\
\ldots \\
\omega \\
\end{vmatrix}
\end{equation*}

\end{definition}

To obtain an S-box $S_{\eta}$ where $S_{\eta}(3) \leq 7$ and which maps 0, 1 and 2 to themselves, we apply $MSB\_SHIFT(S(3))$ to $S$ to create a new S-box, $S_{a}$. If $S_{a}(3) > 7$, we create and apply $MSB\_SHIFT(S_{a}(3))$ to $S_{a}$ to obtain $S_{b}$... Eventually, this procedure will yield $S_{\eta}$ as desired.

We can then use on $S_{\eta}$ the procedure we would have used on $S$ had $S(3)$ been less than or equal to 7.
\end{proof}

\begin{theorem}\label{s_three_final_theorem}
Let $S$ be a bijective S-box mapping 0, 1, 2 to themselves, with differential uniformity $\leq 4$. There exist bijective S-boxes $S_1, S_2, S_3$ affine-equivalent to $S$ which also map 0, 1, and 2 to themselves, such that:
\begin{itemize}
\item $S_1(3) = 5$
\item $S_2(3) = 6$
\item $S_3(3) = 7$
\end{itemize}
\end{theorem}
\begin{proof}
If $S$ is APN, or if $S(3) \neq 3$, we have already obtained the desired result. Let us therefore focus on the case where $S$ is differentially-4-uniform and $S(3)=3$.

Apply a matrix $M_a$ to the S-box outputs mapping 2 to some value $x \geq 4$, but preserving the property that 0 and 1 are mapped to themselves. This will map 3 to $(x \oplus 1)$.

Let the S-box resulting from this be denoted $S_a$. We have $S_a(2)=x$, $S_a(3) = (x \oplus 1)$.

The input difference between $S_{a}^{-1}(2)$ and $S_{a}^{-1}(3)$ cannot be equal to 1. For if it could, we would have (input difference 1, output difference 1) for $((x_1, x_2), (y_1, y_2)) =$

\begin{itemize}
\item $((0, 1), (0, 1))$
\item $((2, 3), (x, x \oplus 1))$
\item $((S_{a}^{-1}(2), S_{a}^{-1}(3)), (2, 3))$
\end{itemize}

contradicting our assertion that $S$ has differential uniformity $\leq 4$.

Apply some matrix $M_b$ to the S-box inputs, mapping 2 to $S_{a}^{-1}(2)$ and 1 to itself. This will map 3 to $(S_{a}^{-1}(2) \oplus 1)$, which we have already shown cannot be $S_{a}^{-1}(3)$.

We can now continue with the procedure described in Theorem \ref{s_three_main_body} to achieve our desired result.
\end{proof}

Without loss of generality, we shall henceforth assume that the preferred value for $S(3)$ was 5.

The procedure now starts to become more complicated. The purpose served by Theorem \ref{fourteen} will not be obvious until we reach Theorem \ref{fifteen}. These two theorems will allow us to construct an equivalent S-box mapping all inputs with Hamming weight 0 or 1 to themselves, and 3 to 3 or 5 (we will assume that 3 is mapped to 5 for an S-box with differential uniformity 2 or 4). We will then demonstrate how, after an S-box mapping 0, 1, 2, and 4 to themselves, and 3 to 3 or 5, has been constructed, assuming $n \geq 4$, we can construct another equivalent S-box as described but with additional restrictions on the values of $S(x)$ for $x$ of the form $2^{i}+1$, using a procedure which will reduce the number of times we need to apply that of the below theorem:

\begin{theorem}\label{fourteen}
Let S be a bijective S-box such that, for some $h \geq 1$,
\begin{itemize}
\item $S(x)=x\ {\forall}\ x \in \{0, 1, 2, \ldots, 2^{h}\}$,
\item $S(3)=3\textrm{ or }5$,
\item $S^{-1}(2^{h+1}) < 2^{h+1}$, and
\item $n \geq (h+2)$.
\end{itemize}
There exist at least $(h+3)$ bijective S-boxes $S_2$ affine-equivalent to $S$ such that
\begin{itemize}
\item $S_2(x)=x\ {\forall}\ x \in \{0, 1, 2, \ldots, 2^{h}\}$,
\item $S_2(3)=S(3)=3\textrm{ or }5$,
\item $S_2^{-1}(2^{h+1}) \geq 2^{h+1}$, and
\item $S_2(x) = M(S(x))$ for some linear bijection $M$.
\end{itemize}
\end{theorem}
\begin{proof}

We note first of all that the procedure described in this proof will not be necessary for $h=1$, since under the circumstances described $S^{-1}(4)$ must be greater than or equal to 4. We can therefore replace the assumption that $h \geq 1$ with the assumption that $h > 1$.

The matrix $M$ is in fact a $CXOR$ matrix as described in Definition \ref{cxor_definition}.
Let all of $M$'s columns except the $(n - (h+1))$st be identical to the corresponding column of the identity matrix. Since they are linearly independent, this is permissible.

This column must be of the form

\begin{equation*}
\begin{vmatrix}
0 \\
\ldots \\
0 \\
1 \\
x_{1} \\
x_{2} \\
\\
\ldots \\
\\
x_{h+1} \\
\end{vmatrix}
\end{equation*}

where at least one $x_{i}$ should be nonzero, and such that the column should not be the same as any S-box outputs $x$ with the property that $S^{-1}(x) < 2^{h+1}$.
From the condition that at least one $x_i$ should be nonzero, we see that there are at most $2^{h+1}-1$ possible columns.
It should be clear that this column is also linearly independent of the other columns in the matrix.

We look at the question of how many $x$ such that $S^{-1}(x) < 2^{h+1}$ could be equal to such a column. There are $2^{h+1}$ values of $x$ such that $S^{-1}(x) < 2^{h+1}$.

$(h+1)$ of these are the values $\{2^0, 2^1, \ldots, 2^h\}$. None of these could equal such a column.

An $(h+2)$nd such $x$ which cannot correspond to the column described is 0.

An $(h+3)$rd such value of $x$ is $S(3)$. (This is easily shown to follow from the fact that $h > 1$.)

We obtain an $(h+4)$th such value by noting that $S^{-1}(2^{h+1}) < 2^{h+1}$, and that this cannot correspond to the column due to the condition that at least one $x_i$ be nonzero.

We see that at most $(2^{h+1}-(h+4))$ of the first $2^{h+1}$ truth table entries can correspond to such columns, leaving at least $(2^{h+1}-1)-(2^{h+1}-(h+4)) = (h+3)$ valid choices for the column of $M$.

How does this work? The transformation $S_2 = M(S(x))$ will not affect any outputs of S which are $< 2^{h+1}$. It will, however, map the S-box output equal to the column to $2^{h+1}$. The fact that this output corresponds to an input greater than or equal to $2^{h+1}$ is what causes the transformation to obtain the desired result.

Each different choice for this column will result in a different value for $S_{2}(S^{-1}(2^{h+1}))$ (and a different value for $S_{2}^{-1}(2^{h+1}))$, hence our statement that at least $(h+3)$ such S-boxes exist.
\end{proof}

\begin{theorem}\label{fifteen}
Let $S$ be a bijective S-box such that
\begin{itemize}
\item $S(x) = x$ for $x \in \{0, 1, 2, \ldots, 2^{h}\}$ for some $h$,
\item $S(3) = 3\textrm{ or }5$,
\item $S^{-1}(2^{h+1}) \geq 2^{h+1}$, and
\item $n \geq (h+2)$.
\end{itemize}
There exists at least one bijective S-box $S_2$ affine-equivalent to $S$ such that $S_2(x) = x$ for all $x \in \{0, 1, 2, 2^{2}, {\ldots}, 2^{h}, 2^{h+1}\}$, $S_2(3) = S(3)$ and $S_{2}(x) = S(M(x))$ for some linear bijection $M$.
\end{theorem}
\begin{proof}
If $S^{-1}(2^{h+1}) = 2^{h+1}$, we do not need to do anything. Otherwise, let $M$ be such that its $(n-(h+1))$st column is equal to $S^{-1}(2^{h+1})$, and such that its $(n-i)$th column is equal to $2^{i}$ (i.e. equal to the corresponding column of the identity matrix) for all $0 \leq i \leq h$. All other columns can be chosen arbitrarily, as long as they are linearly independent of these and of each other.

$M$ will map $2^{h+1}$ to $S^{-1}(2^{h+1})$, and will leave invariant all values $< 2^{h+1}$. Hence $S_{2}$ will map such values to the same outputs as $S$.

Since $S^{-1}(2^{h+1}) \geq 2^{h+1}$, it cannot be in the spanning set of the columns to the right of it, and hence $M$ is a valid linear bijection.
\end{proof}

It follows from the preceding results that:

\begin{corollary}
Every bijective S-box $S$ is affine-equivalent to at least one bijective S-box $S_2$ such that $S_2$ maps all inputs with Hamming weight $< 2$ to themselves, and 3 to either 3 or 5. Furthermore, if $S$ has differential uniformity of 4 or less, we may assume that $S_2(3)=5$.
\end{corollary}

We will now show how a variation on this procedure may be obtained to construct an S-box with the properties described and with further restrictions on the values mapped to by some of the $x$ with Hamming weight 2.

\begin{theorem}\label{twenty_point_one}
Let $S$ be a bijective S-box such that, for some $h \leq (n-2)$, $S(x) = x$ for all $x \in \{0, 2^0, 2^1, {\ldots}, 2^h\}$, $S(3) \leq 5$, $S(2^{i} + 1) \leq (2^{i+2} - 1)$ for all $1 < i < h$, and such that $n \geq 3$.

(The result does also trivially hold for $h=(n-1)$, but for this value of $h$ the below procedure does not need to be applied.)

There exists at least one affine equivalent bijective S-box $S_2$ such that $S_2(x) = x$ for all $x \in \{0, 1, 2, 2^2, {\ldots}, 2^h\}$, $S_{2}(3)=S(3)$, and $S_{2}(2^{i} + 1) \leq (2^{i+2} - 1)$ for all $1 < i \leq h$.
\end{theorem}
\begin{proof}
If $S(2^{h} + 1) \geq 2^{h+2}$, we keep constructing and applying $MSB\_SHIFT(S(2^{h} + 1))$ until this has ceased to be the case.

$MSB\_SHIFT$ is only used if $S(2^{h} + 1) \geq 2^{h+2}$. It has no effect on any S-box output in which the first nonzero bit occurs two places or more later than the first nonzero bit of $S(2^{h} + 1)$. Hence, any S-box output $< 2^{h+1}$ is unaffected. This means that, for all $0 < i \leq h$, all $S(2^i) = 2^i$ are unaffected, $S(0)$ is unaffected, and for all $(0 < i < h)$, all $S(2^{i}+1) \leq (2^{i+2} - 1)$, and so $\leq (2^{h+1} - 1)$, are unaffected.

We note that the use of $MSB\_SHIFT$ in this fashion will eventually result in a situation where $2^{h+1} \leq S(2^{h} + 1) \leq (2^{h+2} -1)$. (Since no S-box output less than $2^{h+1}$ is affected, we know that $2^{h+1}$ is a lower bound for the value of $S(2^{h} + 1)$ at the end of this procedure.)
\end{proof}

Through the use of $CXOR$ on the $MSB\_SHIFT$ed value, we can reduce the upper bound for the $S(2^{i} + 1)$ even further. Let us begin with the specific case of $S(5)$:

\begin{lemma}\label{twenty_point_two_restricted}
Let $S$ be a bijective S-box such that $S(x) = x$ for $x \in \{0, 1, 2, {\ldots}, 4\}$, such that $S(3) \leq 5$ and such that $n \geq 3$.

There exists at least one affine equivalent bijective S-box $S_2$ such that $S_2(x) = x$ $\forall x \in \{0, 1, 2, {\ldots}, 4\}$, $S_2(3) = S(3)$, $S_2(5) \leq 11$, and (if $S_2(5) > 8$), $S^{-1}(8) \geq 8$.
\end{lemma}
\begin{proof}
If $S(5) \leq 11$ and $S^{-1}(8) \geq 8$ already, we do not need to proceed any further.
If $S(5) \leq 7$, again, we do not need to proceed further.
We assume for the rest of the proof that $n \geq 4$, since the result is trivially true for $n=3$.

If $S_2(5) > 11$, we begin by applying the procedure described in the proof of Theorem \ref{twenty_point_one}, so that $8 \leq S_2(5) \leq 15$.

$S_2(5)$ will now be of the form

\begin{equation*}
\begin{vmatrix}
0 \\
\ldots \\
0 \\
1 \\
y_{1} \\
y_{2} \\
y_{3} \\
\end{vmatrix}
\end{equation*}

We wish to use $CXOR(4, a_{1}a_{2}a_{3})$ as described above to 

\begin{itemize}
\item replace $S_{2}(5)$ with a vector of this form such that $y_1 = 0$, and
\item to ensure that the procedure of Theorem \ref{fourteen} will not need to be applied (as otherwise the application of said procedure might undo what we had achieved here.)

(That is, we need to ensure $S_{2}^{-1}(8) \geq 8$).
\end{itemize}

We know that all of $S(0, 1, 2, 3, 4)$ are less than or equal to 7, but do not know if this is the case for $S(6)$ or $S(7)$. This gives us up to two values that may be of the form $1y_{1}ij$ and be such that using the corresponding $y_{1}ij$ as $a_{1}a_{2}a_{3}$ for the $CXOR$ would result in $S^{-1}(8)$ being less than 8. 
A third such value arises from the need not to xor with $S_{2}(5)$ itself.
As there are four $y_{1}ij$ to choose from $(y_{1}00, y_{1}01, y_{1}10, y_{1}11)$, of which only three are potentially problematic, at least one suitable $y_{1}ij$ for the $CXOR$ will always exist.
\end{proof}

Let us now generalise to the remaining $S(2^{i} + 1)$

\begin{theorem}\label{twenty_point_two_full}
Let $S$ be a bijective S-box such that, for some $h \leq (n-2)$ (The result does trivially hold for $h=(n-1)$, but the below procedure is neither necessary nor applicable in said case):
\begin{itemize}
\item $S(x) = x\ {\forall}\ x \in \{0, 1, 2, 2^{2}, {\ldots}, 2^{h}\}$
\item $S(3) \leq 5$
\item $S(2^{i} + 1) \leq (2^{i+2} - 2i - 1)\ \textrm{for all}\ 0 < i < h$,
\item $n \geq 3$.
\end{itemize}

There exists at least one bijective S-box $S_{2}$ affine equivalent to $S$ such that:

\begin{itemize}
\item $S_2(x) = x\ {\forall}\ x \in \{0, 1, 2, 2^2, {\ldots}, 2^h\}$,
\item $S_2(3) \leq 5$,
\item $S_2(2^{i} + 1) = S(2^{i} + 1)$, and hence $S_2(2^{i} + 1) \leq (2^{i+2} - 2i - 1)$, for all $0 < i < h$, and
\item $S_2(2^h + 1) \leq (2^{h+2} - 2h - 1)$.
\end{itemize}
\end{theorem}
\begin{proof}
Consider first of all the special case $2^{0} = 1$. $S(2^{0}+1) = S(2) = 2$.
($2^{0+2} - 2{\cdot}0 - 1) = (4-0-1) = 3$.

Consider also the particular cases $h=1$ and $h=2$. $S(2^{1}+1) = S(3)$, which we have already stated is less than or equal to 5.
$(2^{1+2} - 2{\cdot}1 - 1) = (8-2-1) = 5$.

$S(2^{2}+1) = S(5)$. $2^{2+2} - 2{\cdot}2 - 1 = 11$, and we have already shown that we can always achieve $S(5) \leq 11$.
(in fact, we will later show that $S(5) \leq 10$ can always be achieved for an APN.)

We see therefore that there is no contradiction inherent in the
properties of $S$ as described above. Let us therefore assume from here on that $h \geq 3$.

If $S_{2}(2^{h} + 1) \leq (2^{h+2} - 2h - 1)$ already, we do not need to proceed any further.

Otherwise, we begin by applying the procedure described in the proof of Theorem \ref{twenty_point_one}, so that $2^{h+1} \leq S_2(2^h + 1) \leq (2^{h+2} - 1)$.

$S_{2}(2^{h} + 1)$ will now be of the form

\begin{equation*}
\begin{vmatrix}
0 \\
\ldots \\
0 \\
1 \\
a_{1} \\
a_{2} \\
\ldots \\
a_{h+1} \\
\end{vmatrix}
\end{equation*}

We wish to use $CXOR(h+2, b_{1}b_{2}{\ldots}b_{h+1})$ to 

\begin{itemize}
\item replace this with another vector of this form such that $a_{1}a_{2}{\ldots}a_{h+1} \leq (2^{h+1} - 2h - 1)$, and
\item to ensure that the procedure of Theorem \ref{fourteen} will not need to be applied (as otherwise said procedure might undo all that we had achieved here.)
\end{itemize}

How many other S-box outputs $S(k)$ ($k < 2^{h+1}$) are there such that
$2^{h+1} \leq S_2(k) \leq (2^{h+2} - 1)$? For $k < 8$, there are at most two (only $S(6)$ and $S(7)$ can, at
this point, take values in such a range.) For larger $k$, there are at most $\sum_{j=3}^{h} 2^{j}-2$
such outputs (since $S(2^{j})$ and $S(2^{j}+1)$, for each $j$, either do not at this stage take values in this range
or, for $j=h$, are the value we wish to $CXOR$.)

This gives us a total of $(2 + \sum_{j=3}^{h} 2^{j}-2)$ outputs which may restrict the range of values we can $CXOR$ with (lest they be mapped to $2^{h+1}$). $2^{2}-2 = 2$, so in fact we have $\leq \sum_{j=2}^{h} (2^{j}-2)$ such outputs. In fact, as $2^{1}-2 = 0$, we have less than or equal to

\begin{equation*}
\sum_{j=1}^{h} (2^{j} - 2) \\
= \\
\sum_{j=1}^{h} 2^{j} - 2h \\
= \\
(2^{h+1} - 2) - 2h \\
= \\
2^{h+1} - 2h - 2 \\
\end{equation*}

such outputs.

In the worst-case scenario, these will prevent the $CXOR$s that would have resulted in the $(2^{h+1} - 2h - 2)$ smallest values $> 2^{h+1}$. Hence $S_2(2^h + 1) \leq [(2^{h+1} + 1) + (2^{h+1} - 2h - 2)] = 2^{h+2} - 2h - 1$.
\end{proof}

We can now obtain the following result:

\begin{theorem}\label{two_twenty}
Every bijective S-box $S$ is affine-equivalent to at least one bijective S-box $S_2$ such that $S_2$ maps all inputs with Hamming weight less than 2 to themselves, 3 to 5, 5 to some value $\leq 11$, and all $2^{i}+1\ (3 \leq i \leq (n-1))$ to some value $\leq 2^{i+2} - 2i - 1$.
\end{theorem}
\begin{proof}

We shall begin by addressing the part of the result that states that $S(5) \leq 11$.

\begin{itemize}
\item Consider applications of Theorem \ref{fifteen}'s procedure carried out after the point at which the procedure described in the proof of Lemma \ref{twenty_point_two_restricted} has either been carried out or deemed unnecessary. These will leave invariant the effect of S-box inputs less than $2^{h+1}$, which by that point will always be $\geq 8$. Hence, these will not affect the value of $S(5)$.

\item If the procedure described in the proof of Lemma \ref{twenty_point_two_restricted} was carried out, the procedure described in the proof of Theorem \ref{fourteen} will not be carried out until $2^{h+1} \geq 16$ (in fact, due to the procedures in the proofs of Lemma \ref{twenty_point_two_restricted} and Theorem \ref{twenty_point_two_full}, it might not be carried out until much later or indeed at all). As $11 < 16$, and as it was stated in the proof of Theorem \ref{fourteen} that no S-box outputs $< 2^{h+1}$ would be affected by the procedure, it follows that the value of $S(5)$ remains unaffected.

\item If this procedure was not carried out, and if the reason for this was that $S(5) \leq 7$, then Theorem \ref{fourteen}'s procedure may be carried out when $2^{h+1} = 8$. However, since this is greater than $S(5)$ under these circumstances, the value of $S(5)$ remains unaffected.
\end{itemize}

We thus prove the part of the result pertaining to $S(5)$.

We now need to address the question of any other values of the form $(2^{i}+1)$ that we have caused to be mapped to values $\leq (2^{i+2}-2i-1)$. For $i=(n-1)$, this is always the case, and we need not consider the corresponding $(2^{i}+1)$.
For the remaining values of $i \geq 3$, it is possible that Theorem \ref{fourteen}'s procedure may
have to be applied, if the value of $S(2^{i}+1)$ was too small to require the procedures of Theorem
\ref{twenty_point_one} and Theorem \ref{twenty_point_two_full}. If so, the value of $S(2^{i}+1)$
will be too small to be affected by application of the matrix $M$ as described, and we already know
that later applications of Theorem \ref{fourteen}'s procedure will leave $S(2^{i}+1)$ unaffected,
just as this application will leave
unaffected $S(2^{j}+1)$ for any $j < i$.

Theorem \ref{fifteen}'s procedure does not affect S-box inputs less than $2^{h+1}$, which exceeds $(2^{i}+1)$.

We thus see how, by incorporating the procedures described in the proofs of Theorem \ref{twenty_point_one} and Lemma \ref{twenty_point_two_restricted} into our construction, we obtain an S-box with the desired properties.
\end{proof}

For an APN or D4U S-box, we can tighten the restrictions on $S(5)$ imposed by the above results further. We begin by
ensuring that $S(5) \neq 7$. Consider the point in the procedure at which we would normally apply
the methodology described in the proof of Lemma \ref{twenty_point_two_restricted}:

\begin{lemma}\label{twentysix}
Let $S$ be a bijective S-box for $n \geq 3$, such that $S$ maps 0, 1, 2 and 4 to themselves, maps 3 to 5, and 5 to 7.

Then there exists at least one S-box $S_{2}$ linear-equivalent to $S$, such that $S_2$ maps 0, 1, 2 and 4
to themselves, 3 to 5, and 5 to:

\begin{itemize}
\item 6, 9, 10, or 11. (if S is APN).
\item 3, 6, 9, 10, or 11. (otherwise).
\end{itemize}
\end{lemma}
\begin{proof}
We currently have an S-box of the form

\begin{equation*}
\begin{pmatrix}
0 & 1 & 2 & 3 & 4 & 5 & 6 & 7 & \ldots \\
0 & 1 & 2 & 5 & 4 & 7 & X_{1} & Y_{1} & \ldots
\end{pmatrix}
\end{equation*}

for two unknown values $(X_{1}, Y_{1})$

We apply a matrix identical to the identity matrix except that the last two rows of its last
two columns are of the form

\begin{equation*}
\begin{vmatrix}
0 & 1 \\
1 & 0
\end{vmatrix}
\end{equation*}

to the inputs, and then to the outputs. (Or the outputs and then the inputs; the order we choose is irrelevant.)

The values with Hamming weight 1 will not map to themselves after the first matrix application;
however they will again after the second, and the transformation will result in an S-box of the form:

\begin{equation*}
\begin{pmatrix}
0 & 1 & 2 & 3 & 4 & 5     & 6 & 7 & \ldots \\
0 & 1 & 2 & 6 & 4 & X_{2} & 7 & Y_{2} & \ldots
\end{pmatrix}
\end{equation*}

Apply $CXOR(3, 11)$ so that 3 will map to 5. The S-box now takes the form:

\begin{equation*}
\begin{pmatrix}
0 & 1 & 2 & 3 & 4 & 5     & 6 & 7 & \ldots \\
0 & 1 & 2 & 5 & 7 & X_{3} & 4 & (Y_{3} = Y_{2}\ \textrm{or}\ Y_2 \oplus 011) & \ldots
\end{pmatrix}
\end{equation*}

Operate on the inputs as described in the proof of Theorem \ref{fifteen}), mapping 4 to $S^{-1}(4)=6$,
so that 4 will again map to itself. We have:

\begin{equation*}
\begin{pmatrix}
0 & 1 & 2 & 3 & 4 & 5	  & 6 & 7 & \ldots \\
0 & 1 & 2 & 5 & 4 & Y_{3} & 7 & X_{3} & \ldots
\end{pmatrix}
\end{equation*}

$Y_{3}$ may be 6 (if so, $Y_{2}$ will have been 5 and $Y_{1}$ will have been 6), or it may be 3 if
the S-box is not APN. If it is neither of these, we can apply the procedure described in the proof of Lemma
\ref{twenty_point_two_restricted} to ensure that 5 maps to 9, 10 or 11
without changing the values mapped to by 0, 1, 2, 3 and 4, and as explained in the proof of
Theorem \ref{two_twenty}, this will remain the case after the full procedure has been applied to the S-box.
\end{proof}

If $S$ is APN, we can then go on to ensure 5 will not map to 9:

\begin{lemma}\label{twenty_point_four}
Let $S$ be a bijective APN S-box for $n \geq 3$, such that $S$ maps 0, 1, 2 and 4 to themselves, maps 3 to 5, and 5 to 6, 9, 10, or 11.

Then there exists at least one S-box $S_{2}$ linear-equivalent to $S$, such that $S_2$ maps 0, 1, 2 and 4
to themselves, 3 to 5, and 5 to 6, 10, or 11.
\end{lemma}
\begin{proof}
If $S(5) \neq 9$, this is already true, and we do not need to do anything.

Otherwise, $S$ is of the form

\begin{equation*}
\begin{pmatrix}
0 & 1 & 2 & 3 & 4 & 5 & 6 & 7 & \ldots \\
0 & 1 & 2 & 5 & 4 & 9 & Y_{1} & Z_{1} & \ldots
\end{pmatrix}
\end{equation*}

for two unknown values $(Y_{1}, Z_{1})$.

$S(6)$ and $S(7)$ cannot both be $\in \{10, 11\}$, since (input difference 1, output difference 1) already occurs for S-box inputs 0 and 1. It will therefore be possible to apply at least one of $CXOR(4,011)$ or $CXOR(4,010)$ so that 5 maps to either 10 or 11, without creating a situation where $S^{-1}(8) < 8$. As previously stated, we can now continue with the rest of the procedure without affecting the value of $S(5)$.
\end{proof}

However, prior to continuing with the rest of the procedure, we can eliminate 11 from the set of possible values for $S(5)$ if $S$ is APN:

\begin{lemma}
Let $S$ be a bijective APN S-box for $n \geq 3$, such that $S$ maps 0, 1, 2 and 4 to themselves, maps 3 to 5, and 5 to 11.

Then there exists at least one S-box $S_{2}$ linear-equivalent to $S$, such that $S_2$ maps 0, 1, 2 and 4
to themselves, 3 to 5, and 5 to 6 or 10.
\end{lemma}
\begin{proof}
$S$ is of the form

\begin{equation*}
\begin{pmatrix}
0 & 1 & 2 & 3 & 4 & 5  & 6     & 7 & \ldots \\
0 & 1 & 2 & 5 & 4 & 11 & X_{1} & Y_{1} & \ldots
\end{pmatrix}
\end{equation*}

$Y_{1}$ cannot be equal to 9, otherwise we would have (input difference 2, output difference 2) for inputs 0, 2, 5, 7, and hence $S$ would not be APN.

If $X_{1}$ is not equal to 9, we can simply apply $CXOR(4, 001)$ to the outputs.

If $X_{1}$ is equal to 9, we cannot do this, as otherwise it will result in the procedure of Result 14
having to be applied when ensuring $S(8)=8$, undoing what we have achieved here. We have

\begin{equation*}
\begin{pmatrix}
0 & 1 & 2 & 3 & 4 & 5  & 6 & 7 & \ldots \\
0 & 1 & 2 & 5 & 4 & 11 & 9 & Y_{1} & \ldots
\end{pmatrix}
\end{equation*}

(Note that $Y_{1} \neq 8$ if $S$ is APN.) Carry out the affine transformation which xors all inputs with 1, resulting in:

\begin{equation*}
\begin{pmatrix}
0 & 1 & 2 & 3 & 4  & 5 & 6     & 7  \ldots \\
1 & 0 & 5 & 2 & 11 & 4 & Y_{1} & 9  \ldots
\end{pmatrix}
\end{equation*}

Carry out the same transformation on the outputs:

\begin{equation*}
\begin{pmatrix}
0 & 1 & 2 & 3 & 4  & 5 & 6     & 7 & \ldots \\
0 & 1 & 4 & 3 & 10 & 5 & Y_{2} & 8 & \ldots
\end{pmatrix}
\end{equation*}

Apply $CXOR(2, 1)$ to the inputs:

\begin{equation*}
\begin{pmatrix}
0 & 1 & 2 & 3 & 4  & 5 & 6 & 7     & \ldots \\
0 & 1 & 3 & 4 & 10 & 5 & 8 & Y_{2} & \ldots
\end{pmatrix}
\end{equation*}

Use the same $CXOR$ on the outputs:

\begin{equation*}
\begin{pmatrix}
0 & 1 & 2 & 3 & 4  & 5 & 6 & 7     & \ldots \\
0 & 1 & 2 & 4 & 11 & 5 & 8 & Y_{3} & \ldots
\end{pmatrix}
\end{equation*}

(Note that since $S$ is APN, $Y_{3} \neq 9$.)

Apply $CXOR(3, 01)$ to the outputs:

\begin{equation*}
\begin{pmatrix}
0 & 1 & 2 & 3 & 4  & 5 & 6 & 7     & \ldots \\
0 & 1 & 2 & 5 & 11 & 4 & 8 & Y_{4} & \ldots
\end{pmatrix}
\end{equation*}

(Since $S$ is APN, $Y_{4} \neq 9$.)

Apply the same $CXOR$ to the inputs:

\begin{equation*}
\begin{pmatrix}
0 & 1 & 2 & 3 & 4 & 5  & 6     & 7 & \ldots \\
0 & 1 & 2 & 5 & 4 & 11 & Y_{4} & 8 & \ldots
\end{pmatrix}
\end{equation*}

We reiterate that $Y_{4}$ cannot be equal to 9, or else we would have (input difference 1, output difference 1) for
((0, 0), (1, 1)) and ((6, 9), (7, 8)).

We now apply $CXOR(4, 001)$ to the outputs, obtaining:

\begin{equation*}
\begin{pmatrix}
0 & 1 & 2 & 3 & 4 & 5  & 6     & 7 & \ldots \\
0 & 1 & 2 & 5 & 4 & 10 & Y_{5} & 9 & \ldots
\end{pmatrix}
\end{equation*}

Note that since $Y_{4}$ could not equal 9, $Y_{5}$ cannot equal 8. As a result, we will not have to apply
the procedure of \ref{fourteen}, and the value of $S(5)$ will not be affected by any subsequent part of
the overall procedure.
\end{proof}

From this, it follows that:

\begin{corollary}
Every bijective S-box $S$ with differential uniformity $\leq 4$ is affine-equivalent to at least one bijective S-box $S_2$ such that $S_2$ maps all inputs with Hamming weight less than 2 to themselves, 3 to 5, 5 to
\begin{itemize}
\item 6 or 10 (if $S$ is APN)
\item 3, 6, 9, 10, or 11 (otherwise)
\end{itemize}
and all $2^{i}+1 (3 \leq i \leq (n-1))$ to some value $\leq 2^{i+2} - 2i - 1$.
\end{corollary}

The question arises as to whether this representation is unique. We have tested this on various S-boxes by generating several affine-equivalent boxes and applying the procedure described above, and unfortunately the representation as described above (whether it be the more general result, the result assuming D4U, or the more restricted representation for APN S-boxes) has not in any of these cases been unique.

\subsubsection{An additional result that applies for APN S-boxes.}

\begin{lemma}\label{result_seventeen}
Let $S$ be a bijective, APN S-box mapping all inputs with Hamming weight 1 to themselves. $S$ does not map any inputs with Hamming weight 2 or 3 to themselves.
\end{lemma}
\begin{proof}
For some value $a \in GF(2^{n})$, let $HW(a)$ denote the Hamming weight of $a$.
If some $x$ with weight 2 was mapped to itself, input difference $x$, output difference $x$ would occur for the two pairs $(0, x)$ and $(x, 0)$.
However, it would also occur for $(2^k, 2^l)$ and $(2^l, 2^k)$ where $(2^k \oplus 2^l) = x$.
This would imply that $S$ had differential uniformity $\geq 4$, contradicting our assertion that it was APN.

If some $x$ with weight 3 was mapped to itself, let $y$ be some input such that $HW(y) = 1$ and $HW(x \oplus y) = 2$. It would follow that $HW(S(x) \oplus S(y)) = 2$. Let $z$ be the value $(x \oplus y)$. From $z$ having weight 2, we see that (input difference $z$, output difference $z$) would occur not only for the two input pairs $(x, y), (y, x)$, but also for $(z_1, z_2), (z_2, z_1)$ where $z_1$ and $z_2$ are the two values with Hamming weight 1 that would xor to give $z$. Again, this contradicts the assumption of almost-perfect nonlinearity.
\end{proof}

\subsection{Consequences for genetic/memetic and ant algorithms}\label{SectionConseqThree}
\subsubsection{Genetic and memetic algorithms}
The concept of \textit{epistasis} in genetic algorithms is introduced in various tutorials on the subject \cite{Whitley1994, Townsend2003}, and is also relevant in the context of memetic algorithms \cite{Cotta_Neri2012}. It is stated that, for such algorithms to be effective, the representation of the entities being evolved should be such that there is ``little interaction between genes''. In the context of Townsend's tutorial \cite{Townsend2003}, the entities were strings of bits, and the individual genes were the individual bits. As far as possible, Townsend indicated, the effect of one bit's value on the fitness of the candidate solution should be independent of any other bit's effect.

While no rigorous mathematical definition based on the level of dependence was given, the representation of the candidate solutions was deemed to have high epistasis if the effects of certain bits on the fitness were in fact highly dependent on the values of other bits. This is not desirable.

In the case we are dealing with (bijections over $GF(2^{n})$), the genes are the positions assigned to the individual output values. Clearly, if we use the unmodified truth table as the representation, the level of epistasis is extremely high. No individual output value contains any information, since given any S-box $S$ with $S(x)=y$ and a given set of affine invariant properties, we may xor the set of outputs with any nonzero $m$-bit value to obtain $S_2$ with $S_{2}(x) \neq y$ and precisely the same set of properties. Even if we ensure $S(0)=0$ to prevent this, using a different matrix to that defined in Theorem \ref{theorem_7_1} means that $(S(1), S(2))$ can still be mapped to any pair of arbitrary nonzero values. Clearly no information on the overall fitness exists until after we have fixed the values for these three truth table entries and at least one other - possibly not even then - and the effect of $S(x)$'s value for any $x > 2$ is clearly dependent on the values for these three and other entries.

We have attempted to counter this in our genetic and memetic algorithms by using the results of Section \ref{construction_main_section} to fix the values of all $S(2^{i})$ and $S(3)$, and to enforce the stated restrictions on the values of the $S(2^{i}+1)$ so that these are consistent with the restrictions that would apply for an APN. Unfortunately, in experiments using the entries in the difference distribution table to calculate fitness, this has led to poorer average fitness than when we have not done so, for the following reasons:

\begin{enumerate}
\item In the experiments, the mutation phase must either block mutations that would violate these restrictions, or allow them but then execute the full procedure of Section \ref{construction_main_section} to render the candidate compliant again.
\begin{itemize}
\item In the first case, this removed potential improvements from the set of mutations that would otherwise have been beneficial in terms of fitness.
\item In the second case, the number of changes in truth table entries resulting from the transformation procedure appears to have impaired the algorithm's convergence.
\end{itemize}
\item The impaired convergence may well result from the non-uniqueness of the representation described. As an example, we evolved an APN S-box for $n=5$ and applied an algorithm that generated S-boxes affine-equivalent to it and then applied the procedure of Section \ref{construction_main_section} to these until it was unable to find any more equivalent boxes satisfying these restrictions.

In this experiment, the value of $S(6)$ took on every value $\in [17, 31]{\backslash}\{24\}$ in at least one of the thus-constructed S-boxes. $S(31)$ took on twenty different values, and the values in other positions were similarly varied. It would appear that the level of epistasis is still extremely high even after the affine transformations are used to fix and restrict truth table values as defined in Section \ref{construction_main_section}, and without knowledge of further transformations which could tighten the restrictions further, it is not clear how this can be addressed. Neither is there any alternate representation for affine equivalence group equivalence classes which would possess less epistasis, and all known alternate representations for bijections over $GF(2^{n})$ are themselves in one-to-one correspondence with different truth tables.
\item The effects as described affected not only the mutation phase, but also the local optimisation phase when memetic algorithms were used. Since the local optimisation phase involved the equivalent of several mutations per individual at each generation, the loss of performance was even more marked.
\end{enumerate}

In addition to the above, there exist more general notions of equivalence than affine-equivalence: EA-equivalence (Definition \ref{EA_equivalence_defn}) and CCZ-equivalence (Definition \ref{CCZ_equivalence_defn}). Both of these generalise affine equivalence, and their invariant properties include the frequency with which each DDT entry occurs, as also the frequencies of absolute values of entries in the cryptographically relevant autocorrelation table and linear approximation table (the latter being particularly important in calculating the crucial nonlinearity property.) This fact almost certainly serves to increase the level of epistasis further, as non affine-equivalent S-boxes may still be CCZ or EA-equivalent, and possess identical fitness in terms of these tables. At this time, we do not know how to exploit the CCZ and EA transformations to impose further restrictions on output values that might counter this, and are forced to leave this as a matter for future research.

Finally, we note that experiments were made in using the restrictions in conjunction with simulated annealing, by preventing moves that would violate the constraints. The loss in performance observed was consistent with that observed for the memetic algorithms as noted.

\subsubsection{Ant algorithms}
In experiments for $n=5$, the restrictions as described again led to worse performance for all parameter choices for the Ant System \cite{Colorni_Dorigo_Maniezzo1996} and for two versions of Ant Colony System \cite{Dorigo_Gambardella1997, Luke2009Metaheuristics}. Why this was so is not so clear as in the case of memetic algorithms, although local optimisation used in these algorithms would have been affected for the same reasons (either due to restricted moves, or retransformations impairing convergence.)
\end{appendices}
\end{document}